\documentclass[onecolumn, prx, superscriptaddress,notitlepage, nofootinbib]{revtex4-2}
\usepackage{graphicx}%
\usepackage{dcolumn}%
\usepackage{bm}%
\usepackage[usenames,dvipsnames]{color}
\usepackage{amsthm}
\usepackage{enumitem} 

\newcolumntype{P}[1]{>{\centering\arraybackslash}p{#1}}
\usepackage{mathtools}
\usepackage[most]{tcolorbox}
\usepackage{multirow}
\usepackage{gensymb}
\usepackage[normalem]{ulem}
\usepackage{CJK}
\usepackage{comment}
\usepackage[colorlinks, linkcolor=blue,anchorcolor=blue,citecolor=blue,urlcolor=blue]{hyperref}
\usepackage{amssymb}
\usepackage{pifont}
\usepackage{physics}
\usepackage{natbib}
\newtheorem{definition}{Definition}
\newtheorem{theorem}{Theorem}
\usepackage{float} 
\makeatletter 
\let\newfloat\newfloat@ltx 
\makeatother
\usepackage{algorithm}
\usepackage{algpseudocode}
\usepackage{caption}
\usepackage{subcaption}
\usepackage{amsmath}
\newcommand{\g}{\mathfrak{g}}

\usepackage{booktabs}
\usepackage{colortbl}
\usepackage{caption}
\newtheorem{fact}{Fact}[theorem]

\newtheorem{corollary}[theorem]{\textbf{Corollary}}

\begin{document}
\title{Prospects of Privacy Advantage in Quantum Machine Learning} 
\author{Jamie Heredge}
\affiliation{Global Technology Applied Research, JPMorgan Chase, New York, NY 10017}
\affiliation{School of Physics, The University of Melbourne, Parkville, VIC 3010, Australia}
\author{Niraj Kumar}
\email{niraj.x7.kumar@jpmchase.com}
\affiliation{Global Technology Applied Research, JPMorgan Chase, New York, NY 10017}
\author{Dylan Herman}
\affiliation{Global Technology Applied Research, JPMorgan Chase, New York, NY 10017}
\author{Shouvanik Chakrabarti}
\affiliation{Global Technology Applied Research, JPMorgan Chase, New York, NY 10017}
\author{Romina Yalovetzky}
\affiliation{Global Technology Applied Research, JPMorgan Chase, New York, NY 10017}
\author{Shree Hari Sureshbabu}
\affiliation{Global Technology Applied Research, JPMorgan Chase, New York, NY 10017}
\author{Changhao Li}
\affiliation{Global Technology Applied Research, JPMorgan Chase, New York, NY 10017}
\author{Marco Pistoia}
\affiliation{Global Technology Applied Research, JPMorgan Chase, New York, NY 10017}
\date{\today}

\begin{abstract}
    
Ensuring data privacy in machine learning models is critical, particularly in distributed settings where model gradients are typically shared among multiple parties to allow collaborative learning. Motivated by the increasing success of recovering input data from the gradients of classical models, this study addresses a central question: \emph{How hard is it to recover the input data from the gradients of quantum machine learning models?} Focusing on variational quantum circuits (VQC) as learning models, we uncover the crucial role played by the dynamical Lie algebra (DLA) of the VQC ansatz in determining privacy vulnerabilities. While the DLA has previously been linked to the classical simulatability and trainability of VQC models, this work, for the first time, establishes its connection to the privacy of VQC models. In particular, we show that properties conducive to the trainability of VQCs, such as a polynomial-sized DLA, also facilitate the extraction of detailed \emph{snapshots} of the input. We term this a weak privacy breach, as the snapshots enable training VQC models for distinct learning tasks without direct access to the original input. Further, we investigate the conditions for a strong privacy breach where the original input data can be recovered from these snapshots by classical or quantum-assisted polynomial time methods. We establish conditions on the encoding map such as classical simulatability, overlap with DLA basis, and its Fourier frequency characteristics that enable such a privacy breach of VQC models. Our findings thus play a crucial role in detailing the prospects of quantum privacy advantage by guiding the requirements for designing quantum machine learning models that balance trainability with robust privacy protection.
\end{abstract}

\maketitle

\section{Introduction}

In the contemporary technological landscape, data privacy concerns command increasing attention, particularly within the domain of machine learning (ML) models that are trained on sensitive datasets. Privacy concerns are widespread in many different applications, including financial records \cite{liu2023efficient, awosika2023transparency}, healthcare information \cite{Kaissis2020SecurePA, Ahamed2023, Aguiar2023}, and location data \cite{park2023fedgeo}, each providing unique considerations. Furthermore, the multi-national adoption of stringent legal frameworks \cite{Albrecht2016HowTG} has further amplified the urgency to improve data privacy.

The introduction of distributed learning frameworks, such as federated learning \cite{Brauneck2023, Tong2022, mcmahan2023communicationefficient}, not only promises increased computational efficiency but also demonstrates the potential for increased privacy in ML tasks. In federated learning, each user trains a machine learning model, typically a neural network, locally on their device using their confidential data, meaning that they only need to send their model gradients to the central server, which aggregates gradients of all users to calculate the model parameters for the next training step. As the user does not send their confidential data, but rather their training gradients, this was proposed as the first solution to enable collaborative learning while preventing data leakage. However, subsequent works have shown that neural networks are particularly susceptible to gradient inversion-based attacks to recover the original input data \cite{Zhu19, huang2021evaluating, zhao2020idlg, GeipingBD020, Yin21}. To mitigate the above issue, classical techniques have been proposed to enhance the privacy of distributed learning models, ranging from gradient encryption-based methods \cite{Phong2018}, the addition of artificial noise in the gradients to leverage differential-privacy type techniques \cite{mcmahan2023communicationefficient}, or strategies involving the use of batch training to perform gradient mixing \cite{eloul2022enhancing}. These techniques, although mitagative in nature, are not fully robust since they either still leak some input information, add substantial computational overhead while training the model in the distributed setting, or result in reduced performance of the model. 

A natural question that follows is whether quantum machine learning can help mitigate the privacy concerns that their classical counterparts exhibit. Specifically, one is interested in exploring the fundamental question underpinning the privacy of quantum models: \emph{Given the gradients of a quantum machine learning model, how difficult is it to reconstruct the original classical data inputs?} In search of quantum privacy advantages, several quantum distributed learning proposals have been previously introduced \cite{Huang2022QuantumFL, Qi2023OptimizingQF, Chehimi2021QuantumFL, li2023blind, Gurung2023DecentralizedQF, Lusnig_2024, gilboa2023exponential, li2024privacy, koyasu2023distributed}. Previous methods for improving privacy in a federated learning context have ranged from the use of blind quantum computing \cite{Li2021QuantumFL}, high-frequency encoding circuits \cite{kumar2023expressive}, and hybrid quantum-classical methods that combine pre-trained classical models with quantum neural networks \cite{chen2021}. In particular, the work of \cite{kumar2023expressive} considered variational quantum circuits (VQC) as quantum machine learning models and suggested that highly expressive product encoding maps along with an overparameterized hardware efficient ansatz (HEA) would necessitate an exponential amount of resources (in terms of the number of qubits $n$) for an attacker to learn the input from the gradients. Their work, although the first and sole one to date to theoretically analyze the privacy of a specific VQC model architecture, has certain key drawbacks. The first is that overparameterization of a HEA leads to an untrainable model, since it mixes very quickly to a 2-design \cite{haah2024efficient} and thus leads to a barren plateau phenomenon \cite{mcclean2018barren}. The authors enforced the requirement of overparameterization to ensure that there are no spurious local minima in the optimization landscape and that all local minima are exponentially concentrated toward global minima \cite{anschuetz2022}. However, this requires the HEA to have an exponential depth and thus an exponential number of parameters, which precludes efficient training due to an exponential memory requirement to store and update the parameters. Secondly, the difficulty of inverting gradients to recover data primarily stems from the high expressivity, characterized in this case by an exponentially large number of non-degenerate frequencies of the generator Hamiltonian of the encoding map. Introducing high-frequency terms in the encoding map may not be an exclusive quantum effect, as classical machine learning models could also be enhanced by initially loading the data with these high-frequency feature maps \cite{tancik2020fourier}.

While previous studies have aimed to highlight the advantages of employing VQC models in safeguarding input privacy, none have convincingly addressed what sets VQC models apart from classical neural networks in their potential to provide a quantum privacy advantage. A critical aspect missing in a comprehensive examination of the privacy benefits offered by VQC models in a privacy framework tailored for them. Such a framework should avoid dependence on specific privacy-enhancing procedures or architectures and instead focus on exploring the fundamental properties of VQC models that result in input privacy.

\begin{table}[t!]
\centering
\begin{tabular}{|>{\raggedright\arraybackslash}p{2.6cm}||>{\raggedright\arraybackslash}P{4cm}|>{\raggedright\arraybackslash}P{3cm}|>{\raggedright\arraybackslash}P{7cm} |}
\hline
 \textbf{Privacy Breach}& \textbf{Description}& \textbf{Complexity} & \textbf{Requirements}\\
\hline
\hline
 Weak  & Snapshot recovery (Sec~\ref{sec:sanpshot_recovery})& Algorithm~\ref{alg:snapshot_recov}:  $\mathcal{O}(\text{poly}(\text{dim}(\g)))$& $\mathcal{O}(\text{poly}(n))$ sized DLA + LASA condition (Def~\ref{def:lasa})+ Slow Pauli Expansion (Def~\ref{def:slow_pe}) \\
\hline
 Strong  & Snapshot inversion for local Pauli encoding (Sec~\ref{sec:general-pauli})& Algorithm~\ref{alg:snapshot_inver_gen_pauli}: $\mathcal{O}(\text{poly}(n, 1/\epsilon)$ &
 Snapshot recovery requirement + Separable state with $\rho_J(\mathbf{x})$ parameterized by subset $\mathsf{x}_J \subseteq \mathbf{x}$
 \begin{itemize}[noitemsep, topsep=0pt, left=-1pt]
     \item $\textup{dim}(\mathbf{x}_J)  =  \mathcal{O}(1)$
     \item each $x_k$ is encoded at most $R=\mathcal{O}(\textup{poly}(n))$ times
     \item Snapshot components with non-zero overlap w.r.t. $\rho_J(\mathbf{x}_J)$ has cardinality at least $\textup{dim}(\mathbf{x}_J)$.
 \end{itemize} \\
\hline
 \multirow{3}{2.6em}{\raggedright Strong } & Snapshot inversion for generic encoding (Sec~\ref{sec:approx-inversion}) & Grid Search : $\mathcal{O}\left(\left(\frac{L}{\epsilon}\right)^d \right)$ & Privacy breach is not possible \\ 
\hline
\end{tabular}
\caption{Summary of results on the privacy guarantees and complexity provided by studied attack models on various VQC models. We consider two privacy breach scenarios involving VQCs : \emph{weak privacy} breach and \emph{strong privacy} breach for classical or quantum-assisted polynomial time methods. Weak privacy breach concerns the recovery of the meaningful snapshots of the input encoded state, allowing training VQC models for distinct learning tasks without requiring access to the input. Strong privacy breach concerns subsequently inverting the snapshots to recover the original input. We consider the snapshot invertibility for local Pauli encoding map which admits an efficient (polynomial in the number of qubit $n$) algorithm if the requirements stated in the table is met. For the case of generic encoding maps where the VQC is considered as a black-box $L$-Lipschitz function, snapshot invertibility requires performing the grid search which scales exponentially in the input dimension $d$, and thus it rules out efficient privacy breaches.}
\label{tab:summary2}
\end{table}

To address the above concerns, we introduce a framework designed to assess the possibility of retrieving classical inputs from the gradients observed in VQC models. We consider VQCs that satisfy the Lie algebra supported ansatz (LASA) property, which has been key in establishing connections with the trainability and classical simulatability of VQCs \cite{fontana2023adjoint, somma2004nature, goh2023lie}.  Our study systematically differentiates the separate prerequisites for input reconstruction across both the variational ansatz and encoding map architectures of these VQC models as summarized in Table~\ref{tab:summary2}. Our first result concerns the properties of variational ansatz and the measurement operator of the VQC. Specifically, we show that when the VQC satisfies the LASA condition, i.e., when the measurement operator is within the dynamical Lie algebra (DLA) of the ansatz, and when the DLA scales polynomially with the number of qubits, it is possible to efficiently extract meaningful \emph{snapshots} of the input, enabling training and evaluation of VQC models for other learning tasks without having direct access to the original input. We call this the \emph{weak privacy} breach of the model. Further, we investigate conditions for \emph{strong privacy} breach, i.e., recoverability of the original input by classical or quantum-assisted polynomial time methods. Fully reconstructing the input data from these snapshots to perform a strong privacy breach presents a further challenge which we show is dependent on properties of the encoding map, such as the hardness of classically simulating the encoding, overlap of the DLA basis with encoding circuit generators, and its Fourier frequency characteristics. The two types of privacy breach we introduce are summarised in Figure~\ref{fig:privacy-catagories} while more specific definitions regarding snapshots, recoverability, and invertibility are provided in Section~\ref{sec:input_recover}.

This investigation presents a comprehensive picture of strategies to extract the key properties of VQCs to provide robust privacy guarantees while ensuring that they are still trainable. We structure our paper in the following manner. Appendix~\ref{sec:notations} provides the notation used in this work. Sec~\ref{sec:general_framework} provides a general framework for studying privacy with VQC. This includes describing the VQC framework, providing Lie theoretic definitions required for this work, and the privacy definitions in terms of input recoverability. Next, Sec~\ref{sec:sanpshot_recovery} followed by \ref{sec:invertibility} covers a detailed analysis of the snapshot recoverability from the gradients, and snapshot inversion to recover the input, respectively. Sec~\ref{sec:train_privacy} establishes the connections between privacy and the well-studied trainability of VQCs, and Sec~\ref{sec:future_dir} consequently highlights the future directions of enabling robust privacy with quantum machine learning models. We finally conclude our results in Sec~\ref{sec:conclusion}.

\begin{figure}[t!]
    \centering
    \includegraphics[width=0.9\textwidth]{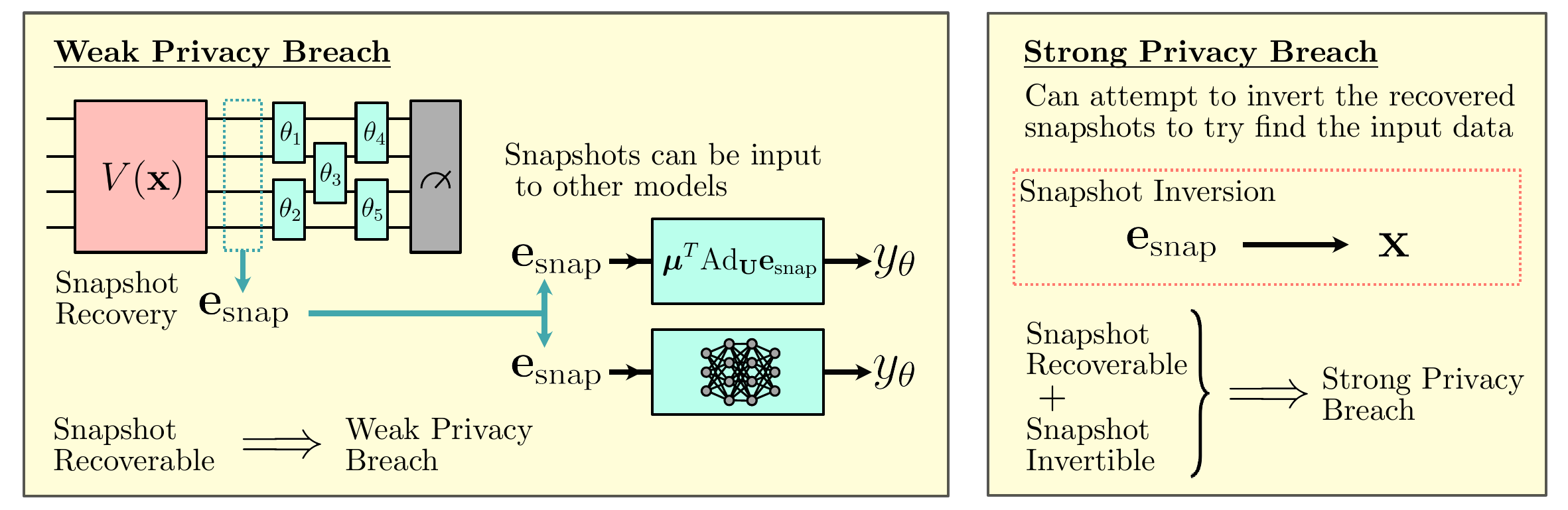}
    \caption{Overview of the general framework and definitions. Weak privacy breach corresponds to attacks where snapshots of the data are retrieved. These can be used as inputs to other models, without explicitly needing the exact data, allowing one to potentially learn characteristics of the data. If these snapshots can then be further inverted to retrieve the input data $\mathbf{x}$ explicitly, we say the attack has succeeded in a strong privacy breach.}
    \label{fig:privacy-catagories}
\end{figure}

\section{General Framework} \label{sec:general_framework}

\subsection{Variational Quantum Circuits for Machine Learning}

A variational quantum circuit (VQC) is described in the following manner. We consider the $d$ dimensional input vector $\mathbf{x} \in \mathcal{X} \subset \mathbb{R}^d$, which is loaded into the quantum encoding circuit $V(\mathbf{x})$ of $n$ qubits to produce a feature map with the input state mapping,
\begin{equation}
    \rho(\mathbf{x}) = V(\mathbf{x})\ket{0}^{\otimes n}\bra{0}^{\otimes n} V(\mathbf{x})^{\dagger}.
    \label{eq:quantum-encoding-def}
\end{equation}
This operation loads the input vector of dimension $d$ to a Hilbert space $\mathcal{H} = (\mathbb{C}^2)^{\otimes n}$ of dimension $\text{dim}(\rho(\mathbf{x})) = 2^n$. We will explicitly consider the scenario where $n = \Theta(d)$, which is a common setting in most existing VQC algorithms and hence the number of qubits in a given algorithm will be of the same order as the input vector dimension $d$. The state $\rho(\mathbf{x})$ is then passed through a variational circuit ansatz $U(\boldsymbol\theta)$ defined as
\begin{equation}
    \mathbf{U}(\boldsymbol{\theta}) = \prod_{k=1}^{D} e^{- i \theta_{k}\mathbf{H}_{k}},
    \label{eq:ansatz}
\end{equation}
which is parameterized by a vector of variational parameters $\boldsymbol{\theta} = [\theta_{1},\cdots, \theta_{D}]$, where $D$ is the total number of variational parameters. Here $\{\mathbf{H}_1,\cdots, \mathbf{H}_{D}\}$ are the set of $D$ Hermitian generators of the circuit $U$. We note that the above structure is quite general since some common ansatz structures such as the hardware efficient ansatz, the quantum alternating operator ansatz, and Hamiltonian variational ansatz among others, are all encapsulated in this framework as highlighted in \cite{larocca2022diagnosing}.

The parameterized state $\rho(\mathbf{x})$ is passed through a variational circuit denoted by $U(\boldsymbol{\theta})$, followed by the measurement of some observable $\mathbf{O} \in \mathcal{H}$. For a given $\boldsymbol{\theta}$, the output of the variational quantum circuit model is expressed as the expectation value of $\mathbf{O}$ with the parameterized state,
\begin{equation}
    y_{\boldsymbol\theta}(\mathbf{x}) = \text{Tr}(\mathbf{U}^\dagger(\boldsymbol\theta) \boldsymbol{O} \mathbf{U}(\boldsymbol\theta) \rho(\mathbf{x})).
    \label{Eq:model_output}
\end{equation}
For the task of optimizing the variational quantum circuits, the model output is fed into the desired cost function $\texttt{Cost}(\boldsymbol{\theta}, \mathbf{x})$ which is subsequently minimized to obtain,
\begin{equation}
        \boldsymbol \theta ^* = \underset{\boldsymbol \theta}{\text{arg min}} \hspace{1mm} \texttt{Cost}(\boldsymbol\theta, \mathbf{x}),
\end{equation}
where $\boldsymbol{\theta}^*$ are the final parameter values after optimization. Typical examples of cost functions include cross-entropy loss, and mean-squared error loss, among others \cite{wang2020comprehensive}. 

The typical optimization procedure involves computing the gradient of the cost function with respect to the parameters $\boldsymbol{\theta}$, which in turn, involves computing the gradient with respect to the model output $y_{\boldsymbol\theta}(\mathbf{x})$
\begin{equation}
    C_j = \frac{\partial y_{\boldsymbol\theta}(\mathbf{x})}{\partial \theta_j}, \hspace{2mm} j \in [D].
    \label{eq:gradient}
\end{equation}
Going forward, we will directly deal with the recoverability of input $\mathbf{x}$ given $C_j$, instead of working with specific cost functions. Details of how to reconstruct our results when considering gradients with respect to specific cost functions are covered in Appendix~\ref{app:loss-function}.

\subsection{Lie Theoretic Framework}

We review some introductory as well as recent results on Lie theoretic framework for variational quantum circuits which are relevant to our work. For a more detailed review of this topic, we refer the reader to \cite{larocca2022diagnosing, ragone2022representation}. We provide the Lie theoretic definitions for a periodic ansatz of the form Eq~\ref{eq:ansatz}.

\begin{definition}[Dynamical Lie Algebra]
    The dynamical Lie algebra (DLA) $\g$ for an ansatz $\mathbf{U}(\boldsymbol{\theta})$ of the form Eq~\ref{eq:ansatz} is defined as the real span of the Lie closure of the generators of $U$
    \begin{equation}
        \g = \text{span}_{\mathbb{R}}\langle i\mathbf{H}_1, \cdots, i\mathbf{H}_D \rangle_{\text{Lie}},
    \end{equation}
    where the closure is defined under taking all possible nested commutators of $S = \{i\mathbf{H}_1,\cdots, i\mathbf{H}_D\}$. In other words, it is the set of elements obtained by taking the commutation between elements of $S$ until no further linearly independent elements are obtained. 
\end{definition}

\begin{definition}[Dynamical Lie Group]
    The dynamical Lie group $\mathcal{G}$ for an ansatz $\mathbf{U}(\boldsymbol{\theta})$ of the form of Eq~\ref{eq:ansatz} is determined by the DLA $\g$ such that,
    \begin{equation}
        \mathcal{G} = e^{\g},
    \end{equation}
    where $e^{\g} := \{e^{i\mathbf{H}}, \hspace{2mm} i\mathbf{H} \in \g\}$ and is a subgroup of $\textup{SU}(2^n)$. For generators in $\g$, the set of all $\mathbf{U}(\boldsymbol{\theta})$ of the form Eq~\ref{eq:ansatz} generates a dense subgroup of $\mathcal{G}$.
\end{definition}

\begin{definition}[Adjoint representation]
The Lie algebra adjoint representation is the following linear action:
 $\forall \mathbf{K}, \mathbf{H} \in \g$,
\begin{align}
\textup{ad}_{\mathbf{H}}\mathbf{K} := [\mathbf{H}, \mathbf{K}] \in \g,
\end{align}
and the Lie group adjoint representation is the following linear action
 $\forall \mathbf{U} \in \mathcal{G}, \forall \mathbf{H} \in \g$,
\begin{align}
\textup{Ad}_{\mathbf{U}}\mathbf{H} := \mathbf{U}^{\dagger}\mathbf{H}\mathbf{U} \in \g.
\end{align}
    
\end{definition}

\begin{definition}[DLA basis]
    The basis of the DLA is denoted as $\{i\mathbf{B}_{\alpha}\}_\alpha$, $\alpha \in \{1,\cdots,\textup{dim}(\g)\}$, where $\mathbf{B}_\alpha$ are Hermitian operators and form an orthonormal basis of $\g$ with respect to the Frobenius inner product. 
\end{definition}

Any observable $\mathbf{O}$ is said to be entirely supported by the DLA whenever $i\mathbf{O} \in \g$, or in other words
\begin{equation}
    \mathbf{O} = \sum_\alpha \mu_\alpha \mathbf{B}_\alpha,
\end{equation}
where $\mu_\alpha$ is coefficient of support of $\mathbf{O}$ in the basis $\mathbf{B}_\alpha$.

\begin{definition}[Lie Algebra Supported Ansatz \cite{fontana2023adjoint}] \label{def:lasa}
    A Lie Algebra Supported Ansatz (LASA) is a periodic ansatz of the form Eq~\ref{eq:ansatz} of a VQC where the measurement operator $\boldsymbol{O}$ is completely supported by the DLA $\g$ associated with the generators of $U(\boldsymbol{\theta})$, that is,
    \begin{equation}
        i\boldsymbol{O} \in \g.
    \end{equation}
\end{definition}
In addition to its connections to the trainability of a VQC, this condition also implies that $\forall \boldsymbol{\theta}, U^{\dagger}(\boldsymbol{\theta})i\boldsymbol{O}U(\boldsymbol{\theta}) \in \g$, which enables us to express the evolution of the observable $\boldsymbol{O}$ in terms of elements of $\g$. This is key to some simulation algorithms that are possible for polynomial-sized DLAs \cite{somma2004nature, goh2023lie}.

\subsection{Input Recoverability Definitions} \label{sec:input_recover}

In this section, we provide meaningful definitions of what it means to recover the classical input data given access to the gradients $\{C_j\}_{j=1}^{D}$ of a VQC. Notably, our definitions are motivated in a manner that allows us to consider the encoding and variational portions of a quantum variational model separately.

A useful concept in machine learning is the creation of data \emph{snapshots}. These snapshots are compact and efficient representations of the input data's feature map encoding. Essentially, a snapshot retains enough information to substitute for the full feature map encoded data, enabling the training of a machine learning model for a distinct task with the same data but without the need to explicitly know the input data was passed through the feature map. For example, in methods such as $\mathfrak{g}$-sim \cite{goh2023lie}, these snapshots are used as input vectors for classical simulators. The simulator can then process these vectors efficiently under certain conditions, recreating the operation of a variational quantum circuit.

It will become useful to classify the process of input data $\mathbf{x}$ recovery into two stages; the first concerns recovering snapshots of the quantum state $\rho(\mathbf{x})$ (Eq~\ref{eq:quantum-encoding-def}) from the gradients, which involves only considering the variational part of the circuit.

\begin{definition}[Snapshot Recovery] \label{def:snap_rec}
    Given the gradients $C_j, \hspace{2mm} j \in [D]$ as defined in Eq~\ref{eq:gradient} as well as the parameters $\boldsymbol{\theta} = [\theta_1, \cdots, \theta_{D}]$, we consider a VQC to be snapshot recoverable if there exists an efficient $\mathcal{O}(poly(d, \frac{1}{\epsilon}))$ classical polynomial time algorithm to recover the vector $\mathbf{e}_{\text{snap}}$ such that,
    \begin{equation}
        |[\mathbf{e}_{\text{snap}}]_\alpha - \textup{Tr}(\mathbf{B}_\alpha \rho(\mathbf{x}))| \leq \epsilon, \hspace{1mm} \forall \alpha \in [\dim(\g)] ,
    \end{equation}
    for some $\{\mathbf{B}_\alpha\}$ forming an  Frobenius-orthonormal basis of the DLA $\g$ corresponding to $U(\boldsymbol{\theta})$ in Eq~\ref{eq:ansatz}, and the above holds for any  $\epsilon > 0$. We call $\mathbf{e}_{\text{snap}}$  the snapshot of $\mathbf{x}$. 
\end{definition} 
    
In other words, $ \mathbf{e}_{\text{snap}}$ is the orthogonal projection of the input state $\rho(\mathbf{x})$ onto the DLA of the ansatz, and thus the elements of $ \mathbf{e}_{\text{snap}}$   are the only components of the input state that contribute to the generation of the model output $y_{\boldsymbol{\theta}}(\mathbf{x})$ as defined in Eq~\ref{Eq:model_output}. Here, we constitute the retrieval of the snapshot $ \mathbf{e}_{\text{snap}}$  of a quantum state $\rho(\mathbf{x})$ as \emph{weak privacy} breach, since the snapshot could be used to train the VQC model for other learning tasks involving the same data $\{\mathbf{x}\}$ but without the need to use the actual data. As an example, consider an adversary that has access to the snapshots corresponding to the data of certain customers. Their task is to train the VQC to learn the distinct behavioral patterns of the customers. It becomes apparent that the adversary can easily carry out this task without ever needing the original data input since the entire contribution of the input $\mathbf{x}$ in the VQC output decision-making $y_{\boldsymbol{\theta}}(\mathbf{x})$ is captured by $ \mathbf{e}_{\text{snap}}$.

Next, we consider the stronger notion of privacy breach in which the input data $\mathbf{x}$ must be fully reconstructed. Assuming that the snapshot has been recovered, the second step we therefore consider is inverting the recovered snapshot $ \mathbf{e}_{\text{snap}}$ to find the original data $\mathbf{x}$, a process that is primarily dependent on the encoding part of the circuit. Within our snapshot inversion definition, we consider two cases that enable different solution strategies, snapshot inversion utilizing purely classical methods, and snapshot inversion methods that can utilize quantum samples.

\begin{definition}[Classically Snapshot Invertible Model]
    Given the snapshot $\mathbf{e}_{\text{snap}}$ as the expectation values of the input state $\rho(\mathbf{x})$, we say that VQC admits classical snapshot invertibility, if there exists an efficient $\mathcal{O}(\text{poly}(d, \frac{1}{\epsilon}))$ polynomial time classical randomized algorithm to recover
    \begin{equation}
        \mathbf{x'} : \| \mathbf{x}' - \mathbf{x}\|_2 \leq \epsilon,
    \end{equation}
    with probability at least $p = \frac{2}{3}$, for any user defined $\epsilon >0$.
\end{definition}

\begin{definition}[Quantum Assisted Snapshot inversion]
    Given the snapshot $\mathbf{e}_{\text{snap}}$ as the expectation values of the input state $\rho(\mathbf{x})$, and the ability to query $\text{poly}(d, \frac{1}{\epsilon})$ number of samples from the encoding circuit $V$ to generate snapshots $\mathbf{e}_{\text{snap}}'$ for any given input $\mathbf{x}'$, we say that VQC admits quantum-assisted snapshot invertibility, if there exists an efficient $\mathcal{O}(\text{poly}(d, \frac{1}{\epsilon}))$ polynomial time classical randomised algorithm to recover
    \begin{equation}
        \mathbf{x'} : \| \mathbf{x}' - \mathbf{x}\|_2 \leq \epsilon,
    \end{equation}
    with probability at least $p = \frac{2}{3}$, for any user defined $\epsilon >0$.
\end{definition}

In this work, we specifically focus on input recoverability by considering the conditions under which VQC would admit snapshot recovery followed by snapshot invertibility.  Considering these two steps individually allows us to delineate the exact mechanisms that contribute to the overall recovery of the input.

It is important to mention that it may potentially only be possible to recover the inputs of a VQC up to some periodicity, such that there only exists a classical polynomial time algorithm to recover $\tilde{\mathbf{x}} = \mathbf{x} + \mathbf{k}\pi$ up to $\epsilon$-closeness, where $\mathbf{k} \in \mathbb{Z}$. As the encodings generated by quantum feature maps inherently contain trigonometric terms, in the most general case it may therefore only be possible to recover $\mathbf{x}$ up to some periodicity. However, this can be relaxed if the quantum feature map is assumed to be injective.

Figure~\ref{fig:g-sim-explanation}. shows a diagram that highlights the Lie algebraic simulation method \cite{goh2023lie} along with specifications of the input recovery framework as defined in this work.

\begin{figure}[t!]
    \centering
    \includegraphics[width=0.95\textwidth]{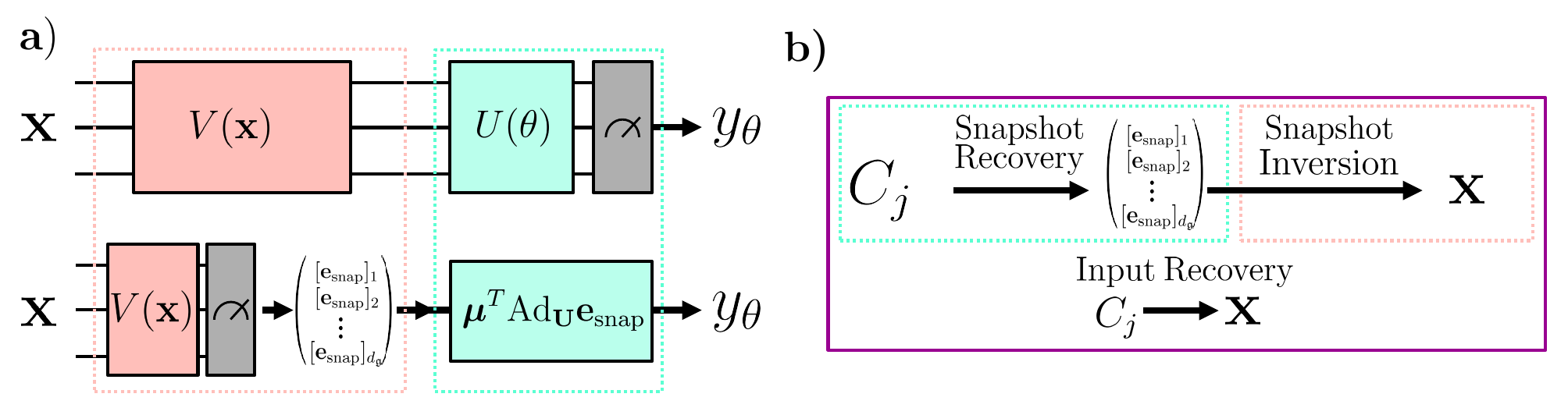}
    \caption{a) Visualization of the difference between the circuit implementation of a variational quantum model and a Lie algebraic simulation procedure of the same model \cite{goh2023lie}. In the upper circuit, a VQC works by encoding the input data $\mathbf{x}$ into a quantum circuit using the encoding step $V(\mathbf{x})$, which is then passed through a variational circuit $U(\theta)$. After this some measurement $O$ is taken of the quantum state in order to calculate the model output $y_\theta$. In the lower circuit, the Lie Algebraic Simulation framework \cite{goh2023lie} is shown; similarly, input data $\mathbf{x}$ is encoded into a quantum circuit using the encoding step $V(\mathbf{x})$, however, the measurements are then performed on this encoded state and used to form a vector of snapshot expectation values. This vector of snapshot expectation values can then be passed as inputs to a classical simulator that uses the adjoint form of $U(\theta)$, which can be performed with resources scaling with the dimension of the DLA formed by the generators of $U(\theta)$. b) In this work, we assess the ability to recover an input $\mathbf{x}$ from gradients $C_j$. This can be broken into two parts: Firstly, the snapshot $\mathbf{e}_{\text{snap}}$ must be recovered from the gradients $C_j$, which corresponds to reversing the Lie Algebraic simulation step. Secondly, the recovered snapshot $\mathbf{e}_{\text{snap}}$ must be inverted to find the original data $\mathbf{x}$, which requires finding the values of $\mathbf{x}$ that when input into $V(\mathbf{x})$ will give the same snapshot values $\mathbf{e}_{\text{snap}}$. If both snapshot recovery and snapshot inversion can be performed, then it admits efficient input recovery.}
    \label{fig:g-sim-explanation}
\end{figure}

\section{Snapshot Recovery} \label{sec:sanpshot_recovery}

This section addresses the \emph{weak privacy} notion of recovering the snapshots of the input as introduced in Def~\ref{def:snap_rec}. As the name implies, the goal here is to recover the vector $\mathbf{e}_{\text{snap}}$ for some Schmidt orthonormal basis $\{\mathbf{B}_{\alpha}\}_{\alpha \in \text{dim}(\g)}$ of the DLA corresponding to the VQC ansatz $\mathbf{U}(\boldsymbol{\theta})$, given that the attacker is provided the following information,

\begin{enumerate}
    \item $D$ gradient information updates $C_j = \frac{\partial y_{\boldsymbol\theta}(\mathbf{x})}{\partial \theta_j}, j\in [D]$ as defined in Eq~\ref{eq:gradient}.
    \item Ansatz architecture $\mathbf{U}(\boldsymbol\theta)$ presented as an ordered sequence of Hermitian generators  $\{\theta_k, \mathbf{H}_{k}\}_{k=1}^{D}$, where $\mathbf{H}_k$ is expressed as a polynomial (in the number of qubits) linear combination of Pauli strings.
    \item Measurement operator $\mathbf{O}$ which satisfies the LASA condition according to Def~\ref{def:lasa} and expressed as a polynomial (in the number of qubits) linear combination of Pauli strings 
\end{enumerate}

Recovering these snapshots will enable an attacker to train the VQC model for other learning tasks that effectively extract the same information from the input states $\rho(\mathbf{x})$ but without the need to use the actual data. The main component of the snapshot recoverability algorithm makes use of the $\g$-sim \cite{somma2004nature, goh2023lie} framework, which we briefly review in the following subsection while also clarifying some previously implicit assumptions, to construct a system of linear equations that can be solved to recover $\mathbf{e}_{\text{snap}}$ as detailed in Algorithm~\ref{alg:snapshot_recov}.

\subsection{Review of Lie-Algebraic Simulation Framework}

We start by reviewing the $\g$-sim framework \cite{somma2004nature,goh2023lie} for classically computing the cost function and gradients of VQCs, when the observable lies in the DLA of the chosen ansatz. Specifically, this framework evolves the expectation values of observables via the adjoint representation. However, a necessary condition for this procedure to be efficient is that the dimension of the DLA ($\text{dim}(\g)$) is only polynomially growing in the number of qubits.

The first step of $\g$-sim consists of building an orthonormal basis for the DLA $\g$ given $(\{\theta_k, \mathbf{H}_k\})_{k=1}^{D}$. Algorithm \ref{alg:dla_basis} presents a well-known procedure to do this. The procedure simply computes pairwise commutators until no new linearly independent elements are found. Given that all operators are expressed in the Pauli basis, the required orthogonal projectors and norm computations performed by Algorithm \ref{alg:dla_basis} can be performed efficiently. If the dimension of DLA is $\mathcal{O}(\text{poly}(n))$, then the iteration complexity, i.e., the number of sets of commutators that we compute, of this procedure is polynomial in $n$. However, an important caveat is that potentially the elements forming our estimation for the DLA basis could have exponential support on the Pauli basis, which is a result of computing new pairwise commutators at each iteration. Thus, for this overall procedure to be efficient, we effectively require that the nested commutators of the generators $\mathbf{H}_k$ do not have exponential support on the Pauli basis.

\begin{definition}[Slow Pauli Expansion] \label{def:slow_pe}
    A set of Hermitian generators $\{\mathbf{H}_1, \dots, \mathbf{H}_D\}$  on $n$-qubits expressed as linear combinations of $\mathcal{O}(\text{poly}(\dim(\g)))$ Pauli strings satisfies the slow Pauli expansion condition if $\forall r \in [\dim \g]$, $[\mathbf{H}_{r}, [\cdots, [\mathbf{H}_{2}, \mathbf{H}_{1}]]]$ can be expressed as a linear combination of $\mathcal{O}(\text{poly}(\dim(\g)))$ Pauli strings.
\end{definition}

In general, it is unclear how strong of an assumption this is, which means that the attacks that we present may not be practical for all VQCs that satisfy the polynomial DLA condition, and thus privacy preservation may still be possible. Also, it does not seem to be possible to apply the $\g$-sim framework without the slow Pauli expansion condition. Lastly, a trivial example of a set of Hermitian generators that satisfies the slow Pauli expansion are those for the quantum compound ansatz discussed in \cite{fontana2023adjoint}.

\begin{algorithm}
\caption{: Finding DLA basis}
\label{alg:dla_basis}
\begin{algorithmic}
\Require Hermitian circuit generators $\{H_1, \dots, H_D\}$, all elements are linear combinations of polynomially-many Pauli strings
\Ensure $\mathcal{A}''' = \{\mathbf{B}_1, \dots, \mathbf{B}_{{\text{dim}(\mathfrak{g})} }\}$ as the basis for the DLA $\g$ %

\begin{enumerate}
    \item Let $\mathcal{A} = \{H_1, \dots, H_D\}$, with all elements represented in the Pauli basis.
    \item Repeat until breaks
        \begin{enumerate}
            \item Compute pairwise commutators of elements of $\mathcal{A}$ into $\mathcal{A}'$
            \item Orthogonally project $\mathcal{A}'$ onto orthogonal complement of $\mathcal{A}$ in $\g$
            \item Set new $\mathcal{A}''$ to be $\mathcal{A}$ plus new orthogonal elements.
            If no new elements, break.
        \end{enumerate}
    \item Perform Gram--Schmidt on $\mathcal{A}$ forming $\mathcal{A}'''$.
    \item Return $\mathcal{A}'''$.
\end{enumerate}
\end{algorithmic}
\end{algorithm}

Given the orthonormal basis $\mathbf{B}_{\alpha}$ for $\g$, under the LASA condition, we can express $\mathbf{O}=\sum_{\alpha \in [\text{dim}(\g)]} \mu_\alpha \mathbf{B}_\alpha$, and hence we can write the output as
\begin{equation}
    y_{\boldsymbol\theta}(\mathbf{x}) = \text{Tr}(\mathbf{U}^\dagger(\boldsymbol\theta) \mathbf{O} \mathbf{U}(\boldsymbol\theta) \rho(\mathbf{x})) = \sum_\alpha \text{Tr}(\mu_\alpha \mathbf{U}^\dagger \mathbf{B}_\alpha \mathbf{U} \rho(\mathbf{x}) ) = \sum_\alpha \text{Tr}(\mu_\alpha \text{Ad}_{\mathbf{U}}( \mathbf{B}_\alpha)\rho(\mathbf{x}) ). 
\end{equation}

In addition, given the form of $\mathbf{U}$, we can express $\text{Ad}_{\mathbf{U}}$ as,
\begin{align}
    \text{Ad}_{\mathbf{U}} = \prod_{k=1}^{D} e^{-\theta_{k}\text{ad}_{i\mathbf{H}_{k}}}.
\end{align}
We can also compute the structure constants for our basis $\mathbf{B}_{\alpha}$, which is the collection of $\dim(\g) \times \dim(\g)$ matrices for the operators $\text{ad}_{i\mathbf{B}_{\alpha}}$. As a result of linearity, we also have the matrix for each $\text{ad}_{i\mathbf{H}}$ for $\mathbf{H} \in \g$ in the basis $\mathbf{B}_{\alpha}$. Then by performing matrix exponentiation and multiplying $\dim(\g) \times \dim(\g)$ we can compute the matrix for $\text{Ad}_{\mathbf{U}}$.

Using the above, the model output may be written,
\begin{align}
&y_\theta = \sum_{\alpha, \beta}\mu_\alpha [\text{Ad}_{\mathbf{U}}]_{\alpha \beta} \text{Tr}( \mathbf{B}_\beta  \rho(\mathbf{x})) = \boldsymbol{\mu}^T\text{Ad}_{\mathbf{U}} \mathbf{e}_{\text{snap}},
\end{align}
where $\mathbf{e}_{\text{snap}}$ is a vector of expectation values of the initial state, i.e. $[\mathbf{e}_{\text{snap}}]_\beta = \text{Tr}[ \mathbf{B}_\beta  \rho(\mathbf{x})].$

Similar to the cost function, the circuit gradient can also be computed via $\g$-sim. Let,
\begin{equation}
    C_j = \frac{\partial y_\theta}{\partial \theta_j} = \boldsymbol{\mu}^T \frac{\partial \text{Ad}_{\mathbf{U}} }{\partial \theta_j} \mathbf{e}_{\text{snap}} =: \chi^{(j)} \cdot  \mathbf{e}_{\text{snap}},
\end{equation}
where the adjoint term differentiated with respect to $\theta_j$ can be written as,
\begin{align}
\frac{\partial \text{Ad}_{\mathbf{U}} }{\partial \theta_{j}} = 
\left[\prod_{k=j}^{D} e^{\theta_{k}\text{ad}_{i\mathbf{H}_{k}}}\right] \text{ad}_{i\mathbf{H}_j} \left[\prod_{k=1}^{j}e^{\theta_{k}\text{ad}_{i\mathbf{H}_k}}\right]. 
\end{align}
The components of $\chi^{(j)}$ can be expressed as,
\begin{equation}
    \chi_\beta^{(j)} = \sum_\alpha \mu_\alpha \left[\frac{\partial \text{Ad}_{\mathbf{U}} }{\partial \theta_j}\right]_{\alpha, \beta},
\end{equation}
allowing $C_j$ terms to be represented in a simplified manner as
\begin{equation}
\label{eqn:grad-in-exp-snap}
    C_j = \sum_{\beta=1}^{{\text{dim}(\g)}} \chi_\beta^{(j)}[ \mathbf{e}_{\text{snap}}]_\beta. 
\end{equation}

 The key feature of this setup is that the matrices and vectors involved have dimension $\textup{dim}(\mathfrak{g})$, therefore for a polynomial-sized DLA, the simulation time will scale polynomially and model outputs can be calculated in polynomial time \cite{goh2023lie}. Specifically, the matrices for each $\text{ad}_{i\mathbf{H}_k}$ in the basis $\{\mathbf{B}_l\}$ and $\text{Ad}_{\mathbf{U}}$ are polynomial in this case.

This Lie-algebraic simulation technique was introduced in order to show efficient methods of simulating LASA circuits with polynomially sized DLA. In this work, we utilize the framework in order to investigate the snapshot recovery of variational quantum algorithms. Based on the above discussion, the proof of the following theorem is self-evident.

\begin{theorem}[Complexity of $\g$-sim]
If ansatz family $\mathbf{U}(\boldsymbol{\theta})$ with an observable $\mathbf{O}$ satisfies both the LASA condition and Slow Pauli Expansion, then the cost function and its gradients can be simulated with complexity $\mathcal{O}(\text{poly}(\dim (\g)))$ using a procedure that at most queries a quantum device a polynomial number of times to compute the $\dim(\g)$-dimensional snapshot vector $\mathbf{e}_{\text{snap}}$. 
\end{theorem}

\subsection{Snapshot Recovery Algorithm}

\begin{algorithm}
\begin{algorithmic}
\Require Observable $\mathbf{O} $ such that $i\mathbf{O} \in \g$, generators $\{\mathbf{H}_k\}_{k=1}^{D}$, ordered sequence $(\{\theta_k, \mathbf{H}_{k}\})_{k=1}^{D}$, and gradients $C_j = \frac{\partial y_{\boldsymbol\theta}(\mathbf{x})}{\partial \theta_j}, j\in [D]$ for some unknown classical input $\mathbf{x}$.
\Ensure Snapshot $\mathbf{e}_{\text{snap}}$ for $\mathbf{x}$

\begin{enumerate}
    \item Run Algorithm \ref{alg:dla_basis} to obtain an orthonormal basis for the DLA $\{\mathbf{B}_\beta\}_{\beta \in [\text{dim}(\g)]}$
    \item For $\beta \in [\text{dim}(\g)]$, compute the $\text{dim}(\g) \times \text{dim}(\g)$ matrix $\text{ad}_{i\mathbf{B}_\beta}$
    \item For $k \in [D]$, compute the coefficients of $\mathbf{H}_k$ in the basis $\{\mathbf{B}_\beta\}_{\beta \in [\text{dim}(\g)]}$, which gives us $\text{ad}_{i\mathbf{H}_k}$
    \item For $k \in [D]$, compute the  $\text{dim}(\g) \times \text{dim}(\g)$ matrix exponential $e^{\theta_{k}\text{ad}{i\mathbf{H}_k}}$
    \item For $j \in [D]$ compute the $\text{dim}(\g) \times \text{dim}(\g)$ matrix
    \begin{align}
        \frac{\partial \text{Ad}_{\mathbf{U}} }{\partial \theta_{j}} = 
\left[\prod_{k=j}^{D} e^{\theta_{k}\text{ad}_{i\mathbf{H}_{k}}}\right] \text{ad}_{i\mathbf{H}_j} \left[\prod_{k=1}^{j}e^{\theta_{k}\text{ad}_{i\mathbf{H}_k}}\right] .
    \end{align}
    \item For $\beta \in [\text{dim}(\g)]$, compute the coefficients $\mu_\beta$ of $\mathbf{O}$ in the basis $\{\mathbf{B}_\beta\}_{\beta \in [\text{dim}(\g)]}$
    \item For $j \in [D], \beta \in [\text{dim}(\g)]$, compute 
    \begin{align}
    \chi_\beta^{(j)} = \sum_\alpha \mu_\alpha \left[\frac{\partial \text{Ad}_{\mathbf{U}} }{\partial \theta_j}\right]_{\alpha, \beta},
    \end{align}
    and construct $D \times \text{dim}(\g)$ matrix $\mathbf{A}$ with $\mathbf{A}_{rs} = \chi_s^{(r)}$.
    \item Solve the following linear system,
    \begin{align}
        [C_1, \dots, C_{D}]^{\mathsf{T}} = \mathbf{A}\mathbf{y},
    \end{align}
    and return $\mathbf{y}$ as the snapshot $\mathbf{e}_{\text{snap}}$.
\end{enumerate}
\end{algorithmic}
\caption{: Snapshot Recovery}\label{alg:snapshot_recov}
\end{algorithm}

With the framework for the $\mathfrak{g}$-sim \cite{goh2023lie} established, we focus on how snapshots $\mathbf{e}_{\text{snap}}$ of the input data can be recovered using the VQC model gradients $C_j$, with the process detailed in Algorithm~\ref{alg:snapshot_recov}. In particular, the form of Eq~\ref{eqn:grad-in-exp-snap} allows a set-up leading to the recovery the snapshot vector $\mathbf{e}_{\text{snap}}$ from the gradients $\{C_j\}_{j=1}^D$, but requires the ability to solve the system of $D$ linear equations given by $\{C_j\}$ with $\text{dim}(\g)$ unknowns $[\mathbf{e}_{\text{snap}}]_{\beta\in \text{dim}(\g)}$. The following theorem formalizes the complexity of recovering the snapshots from the gradients.   

\begin{theorem}[Snapshot Recovery]\label{thm:snaprecoveryscaling}
    Given the requirements specified in Algorithm~\ref{alg:snapshot_recov}, along with the assumption that the number of variational parameters $D \geq \text{dim}(\g)$, where ${\text{dim}(\g)}$ is the dimension of the DLA $\g$, the VQC model admits snapshot $\mathbf{e}_{\text{snap}}$ recovery with complexity scaling as $\mathcal{O}(\text{poly}(\dim (\g)))$. 
\end{theorem}

\begin{proof}
Firstly, we note that given the gradients $C_j$ and parameters $\theta_{j\in [D]}$, the only unknowns are the components of the vector $\mathbf{e}_{\text{snap}}$ of length ${\text{dim}(\mathfrak{g})}$. Therefore, it is necessary to have $\text{dim}(\mathfrak{g})$ equations in total; otherwise, the system of equations would be underdetermined and it would be impossible to find a unique solution. The number of equations depends on the number of gradients and, therefore, the number of variational parameters in the model, hence the requirement that $D \geq \text{dim}(\g)$.

Assuming now that we deal with the case that there are $D \geq {\text{dim}(\mathfrak{g})}$ variational parameters of the VQC model, we can therefore arrive at a determined system of equations. The resulting system of simultaneous equations can be written in a matrix form as,
\begin{equation}
    \begin{pmatrix}
    C_1 \\
    C_2 \\
    \vdots \\
    C_{D}
    \end{pmatrix}_{D \times 1}
    =
    \begin{pmatrix}
    \chi_1^{(1)} & \chi_2^{(1)} & \cdots & \chi_{{\text{dim}(\mathfrak{g})}}^{(1)} \\
    \chi_1^{(2)} & \chi_2^{(2)} & \cdots & \chi_{{\text{dim}(\mathfrak{g})}}^{(2)}\\
    \vdots & \vdots & \ddots & \vdots \\
    \chi^{(D)}_{{\text{dim}(\mathfrak{g})}} & \chi^{(D)}_{{\text{dim}(\mathfrak{g})}} & \cdots & \chi_{{\text{dim}(\mathfrak{g})}}^{(D)} \\
    \end{pmatrix}_{D \times \text{dim}(\g)}
    \begin{pmatrix}
    [\mathbf{e}_{\text{snap}}]_1\\
    [\mathbf{e}_{\text{snap}}]_2\\
    \vdots \\
    [\mathbf{e}_{\text{snap}}]_{{\text{dim}(\mathfrak{g})}}
    \end{pmatrix}_{\text{dim}(\g) \times 1}.
    \label{eqn:matirx-to-invert}
\end{equation}

In order to solve the system of equations highlighted in Eq~\ref{eqn:matirx-to-invert} to obtain $\mathbf{e}_{\text{snap}}$, we first need to compute the coefficients $\{\chi^{(j)}_\beta\}_{j\in [D], \beta \in [\text{dim}(\g)]}$. This can done by the $\g$-sim procedure highlighted in the previous section and in steps 1-7 in Algorithm~\ref{alg:snapshot_recov} with complexity $\mathcal{O}(\text{poly}(\text{dim}(\g)))$. The next step is to solve the system of equations, i.e., step 8 of Algorithm~\ref{alg:snapshot_recov}, which can solved using Gaussian elimination procedure incurring a complexity $\mathcal{O}(\text{dim}(\g)^3)$ \cite{grcar2011mathematicians}. Thus, the overall complexity of recovering the snapshots from the gradients is $\mathcal{O}(\text{poly}(\text{dim}(\g)))$. This completes the proof. 
\end{proof}

In the case that the dimension of DLA is exponentially large ${{\text{dim}(\mathfrak{g})}} = \mathcal{O}(\exp (n))$, then performing snapshot recovery by solving the system of equations would require an exponential number of gradients and thus an exponential number of total trainable parameters $D = \mathcal{O}(\exp (n))$. However, this would require storing an exponential amount of classical data, as even the variational parameter array $\boldsymbol \theta$ would contain $\mathcal{O}(\exp (n))$ many elements and hence this model would already breach the privacy definition which only allows for a polynomial (in $n = \Theta(d)$) time attacker. In addition, the complexity of obtaining the coefficients $\chi^{(j)}_\beta$ and subsequently solving the system of linear equations would also incur a cost exponential in $n$. Hence, for the system of simultaneous equations to be determined, it is required that ${{\text{dim}(\mathfrak{g})}} = \mathcal{O}(\text{poly} (n))$. Under the above requirement, it will also be possible to solve the system of equations in Eq~\ref{eqn:matirx-to-invert} in polynomial time and retrieve the snapshot vector $\mathbf{e}_{\text{snap}}$. Hence, a model is snapshot recoverable if the dimension of the DLA dimension scales polynomially in $d$.

\section{Snapshot Invertibility} \label{sec:invertibility}

We have shown that in the case that the DLA dimension of the VQC is polynomial in the number of qubits $n$ and the slow Pauli expansion condition (Def~\ref{def:slow_pe}) is satisfied, then it is possible to reverse engineer the snapshot vector $\mathbf{e}_{\text{snap}}$ from the gradients. As a result, this breaks the weak-privacy criterion. The next step in terms of privacy analysis is to see if a strong privacy breach can also occur. This is true when it is possible to recover the original data $\mathbf{x}$ that was used in the encoding step to generate the state $\rho(\mathbf{x})$; the expectation values of this state with respect to the DLA basis elements forms the snapshot $\mathbf{e}_{\text{snap}}$. Hence, even if the DLA is polynomial and snapshot recovery allows the discovery of $\mathbf{e}_{\text{snap}}$, there is still the possibility of achieving some input privacy if $\mathbf{e}_{\text{snap}}$ cannot be efficiently inverted to find $\mathbf{x}$. The overall privacy of the VQC model, therefore, depends on both the data encoding and the variational ansatz.

One common condition that is necessary for our approaches to snapshot inversion is the ability to compute the expectation values $\text{Tr}(\rho(\mathbf{x}') \mathbf{B}_k), \hspace{2mm} \forall k \in [\text{dim}(\g)]$ for some guess input $\mathbf{x}'$. This is the main condition that distinguishes between completely classical snapshot inversion and quantum-assisted snapshot inversion. It is well-known that computing expectation values of specific observables is a weaker condition than $\rho(\mathbf{x})$ being classically simulatable \cite{Suzuki_2022}. Hence, it may be possible to classically perform snapshot inversion even if the state $\rho(\mathbf{x})$ overall is hard to classically simulate. In the quantum-assisted case it is always possible to calculate  $\text{Tr}(\rho(\mathbf{x}' \mathbf{B}_k)$ values by taking appropriate measurements of the encoding circuit $V(\mathbf{x}')$ .

In the first subsection, we present inversion attacks that apply to commonly-used feature maps and explicitly make use of knowledge about the locality of encoding circuit. The common theme among these feature maps is that by restricting to only a subset of the inputs, it is possible to express the $\rho(\mathbf{x})$ or expectations there of in a simpler way. The second subsection focuses on arbitrary encoding schemes by viewing the problem as black-box optimization. In general, snapshot inversion can be challenging or intractable even if the snapshots can be efficiently recovered and/or the feature map can be classically simulated. Our focus will be on presenting sufficient conditions for performing snapshot inversion, which leads to suggestions for increasing privacy.

\subsection{Snapshot Inversion for Local Encodings} \label{sec:separable_enc}

For efficiency reasons, it is common to encode components of the input vector $\mathbf{x}$ in local quantum gates, typically just single-qubit rotations. The majority of the circuit complexity is usually either put into the variational part or via non-parameterized entangling gates in the feature map. In this section, we demonstrate attacks to recover components of $\mathbf{x}$, up to periodicity, given snapshot vectors when the feature map encodes each $x_j$ locally. More specifically, we put bounds on the allowed amount of interaction between qubits that are used to encode each $x_j$. In addition, we also require that the number of times the feature map can encode a single $x_j$ be sufficiently small. While the conditions will appear strict, we note that they are satisfied for some commonly used encodings, e.g., the \emph{Pauli product feature map} or \emph{Fourier tower map} \cite{kumar2023expressive} which was previously used in a VQC model that demonstrated resilience to input recovery.

For the Pauli product encoding, we show that a completely classical snapshot inversion attack is possible. An example of a Pauli product encoding is the following:
 \begin{align}
 \label{eqn:pauli_product_encoding}
    \bigotimes_{j_1}^{n}\rho_j(x_j) = \bigotimes_{j_1}^{n}R_{\mathsf{X}}(x_j)|0\rangle\langle0|R_{\mathsf{X}}(-x_j).
\end{align}
where $R_{\mathsf{X}}$ is the parameterized Pauli $\mathsf{X}$ rotation gate. The Fourier tower map is similar to Equation \eqref{eqn:pauli_product_encoding} but utilizes a parallel data reuploading scheme, i.e.
\begin{equation}
\bigotimes_{j=1}^{d}\Big( \bigotimes_{l=1}^{m} R_{\mathsf{X}}( 5^{l-1} x_j) \Big).
\end{equation}
where $n = d m$, with $m$ being the number of qubits used to encode a single dimension of the input.

\subsubsection{Pauli Product Encoding}

The first attack that we present will specifically target Equation \eqref{eqn:pauli_product_encoding}. However, the attack does apply to the Fourier tower map as well. More generally, the procedure applies to any parallel data reuploading schemes of the form:
\begin{align}
\bigotimes_{j=1}^{d}\Big( \bigotimes_{l=1}^{m} R_{\mathsf{X}}(\alpha_l x_j) \Big).
\end{align}
We explicitly utilize Pauli $\mathsf{X}$ rotations, but a similar result holds for $\mathsf{Y}$ or $\mathsf{Z}$.  For Pauli operator $\mathsf{P}$, let $\mathsf{P}_j := i \mathbb{I}^{\otimes(j-1)} \otimes \mathsf{P} \otimes \mathbb{I}^{\otimes(n-j)}$.

\begin{algorithm}
\begin{algorithmic}
\Require Snapshot vector $\mathbf{e}_{\text{snap}}(\mathbf{x})$ of dimension $\dim(\g) = \mathcal{O}(\text{poly}(n))$ corresponding to a basis $(\mathbf{B}_{k})_{k=1}^{\dim(\g)}$ of DLA $\g$. Each $\mathbf{B}_{k}$ is expressed as a linear combination of $\mathcal{O}(\text{poly}(n))$ Pauli strings. Snapshot inversion is being performed for a VQC model that utilized a trainable portion of $\mathbf{U}(\mathbf{\theta})$ with DLA $\g$ and Pauli product encoding Equation \eqref{eqn:pauli_product_encoding}. Index $j \in [d]$, $\epsilon < 1$
\Ensure  An $\epsilon$ estimate of the $j$-th component $x_j$ of the data input $\mathbf{x} \in \mathbb{R}^{d}$ up to periodicity, or output FAILURE.

\If{$i\mathsf{Z}_j \in \g$}
    \State $\alpha \leftarrow 1, \beta \leftarrow 0$
    \State $\mathbf{W} \leftarrow \mathsf{Y}_j$
\ElsIf{$i\mathsf{Y}_j \in \g$}
    \State $\alpha \leftarrow 1, \beta \leftarrow 0$
    \State $\mathbf{W} \leftarrow \mathsf{Y}_j$
\Else 
\State \begin{enumerate}
        \item Determine set of Pauli strings required to span  elements $(i\mathbf{B}_k)_{k=1}^{\dim(\g)}$ and denote the set $\mathcal{P}_{\g}$.
        \item  $\mathcal{P}_{\g} \leftarrow \mathcal{P}_{\g} \cup \{\mathsf{Z}_j, \mathsf{Y}_j\}$, $\lvert \mathcal{P}_{\g}\rvert = \mathcal{O}(\text{poly}(n))$ by assumption. Reduce $\mathcal{P}_{\g}$ to a basis.
        \item Let $\mathbf{C}$ be a $\lvert \mathcal{P}_{\g}\rvert \times \dim(\g)$ matrix whose $k$-th column corresponds to the components of $i\mathbf{B}_k$ in the basis $\mathcal{P}_{\g}$.
        \item Let $\mathbf{A}$ be a $\lvert \mathcal{P}_{\g}\rvert \times 2$ whose first column contains a $1$ in the row corresponding to $\mathsf{Z}_j$ and whose second column contains a $1$ in the row corresponding to $\mathsf{Y}_j$.
        \item Perform a singular value decomposition on $\mathbf{A}^{\mathsf{T}}\mathbf{C}$, and there are at most two nonzero singular values $r_1, r_2$.
    \end{enumerate} 
    \State \If{$r_1 \neq 1$ and $r_2 \neq 1$}
        \State \textbf{return} FAILURE
    \Else
        \State 1. $\mathbf{W} \leftarrow$ singular vector with singular value 1.
        \State 2. Expand $i\mathbf{W}$ in basis $(i\mathsf{Z}_j, i\mathsf{Y}_j)$ record components as $\alpha$ and $\beta$ respectively.
    \EndIf
\EndIf
\begin{enumerate}
    \item Expand $i\mathbf{W}$ in basis $(\mathbf{B}_k)_{k=1}^{\dim(\g)}$, and record components as $\gamma_k$.
    \item Compute \begin{align}
\tilde{x}_j = \cos^{-1}\left[\frac{2}{\text{sign}(\alpha)\sqrt{\alpha^2 + \beta^2}}\sum_{k=1}^{\dim(\g)}\gamma_k[\mathbf{e}_{\text{snap}}]_k \right] - \tan^{-1}(\beta/\alpha).
\end{align}
\item \textbf{return} $\tilde{x}_j$.
\end{enumerate}
\end{algorithmic}
\caption{: Classical Snapshot Inversion for Pauli Product Encoding}
\label{alg:snapshot_inver_pauli}
\end{algorithm}

\begin{theorem}
\label{thm:pauli_feature_resul}
Suppose that the polynomial DLA and slow Pauli expansion (Def~\ref{def:slow_pe}) conditions are satisfied. Also, suppose that we are given a snapshot vector $\mathbf{e}_{\text{snap}}(\mathbf{x})$ for a VQC with trainable portion $\mathbf{U}(\mathbf{\theta})$ with DLA $\g$ and Pauli product feature encoding (Equation \eqref{eqn:pauli_product_encoding}) and the corresponding DLA basis elements $(\mathbf{B}_{k})_{k=1}^{\dim(\g)}$. The classical Algorithm \ref{alg:snapshot_inver_pauli} outputs an $\epsilon$ estimate of  $x_j$, up to periodicity, or outputs \textup{FAILURE}, with time $\mathcal{O}(\text{poly}(n)\log(1/\epsilon))$.
\end{theorem}
\begin{proof}

Steps 1-5 can be performed in $\mathcal{O}(\text{poly}(n))$ classical time due to the polynomial DLA and slow Pauli expansion conditions. The purpose of step 5 is to compute the angles between the linear subspaces $\g$ and $\text{span}_{\mathbb{R}}\{ i\mathsf{Z}_j, i\mathsf{Y}_j\}$. This is to identify if there is any intersection, i.e. if $\exists \alpha, \beta $ such that $\alpha i\mathsf{Z}_j + \beta i\mathsf{Y}_j \in \g$, which is identified by singular values equal to 1. The algorithm cannot proceed if the intersection is trivially empty as the snapshot vector does not provide the required measurement to obtain $x_j$ efficiently with this scheme. From now on, we suppose that such an element was found.

We can without loss of generality just focus on the one-qubit reduced density matrix for $x_j$. In this case, using Bloch sphere representation:
\begin{align}
\rho_j(x_j) &= R_{\mathsf{X}}(x_j)|0\rangle\langle0|R_{\mathsf{X}}(-x_j)= \frac{\mathbb{I} - \sin(x_j)\mathsf{Y} +\cos(x_j)\mathsf{Z}}{2},
\end{align}

so that 
\begin{align}
\text{Tr}([\alpha\mathsf{Z}_j + \beta \mathsf{Y}_j]\rho_j(x_j)) &= \frac{\alpha}{2}\cos(x_j) - \frac{\beta}{2}\sin(x_j)\\
&=\frac{\text{sign}(\alpha)}{2}\sqrt{\alpha^2 + \beta^2}\cos(x_j + \tan^{-1}(\beta/\alpha)).
\end{align}
However, by assumption, $\gamma_k \in \mathbb{R}$ such that
\begin{align}
\alpha i\mathsf{Z}_j + \beta i\mathsf{Y}_j  = \sum_{k=1}^{\dim(\g)}\gamma_i\mathbf{B}_k\implies \text{Tr}([\alpha\mathsf{Z}_j + \beta \mathsf{Y}_j]\rho_j(x_j)) = \sum_{k=1}^{\dim(\g)}\gamma_k[\mathbf{e}_{\text{snap}}]_k.
\end{align}
So to recover $x_j$, we only need to solve:
\begin{align}
\sum_{k=1}^{\dim(\g)}\gamma_k[\mathbf{e}_{\text{snap}}]_k = \frac{\text{sign}(\alpha)}{2}\sqrt{\alpha^2 + \beta^2}\cos(x_j + \tan^{-1}(\beta/\alpha)),
\end{align}
which after rearranging allows the recovery of 
\begin{align}
x_j = \cos^{-1}\left[\frac{2}{\text{sign}(\alpha)\sqrt{\alpha^2 + \beta^2}}\sum_{k=1}^{\dim(\g)}\gamma_k[\mathbf{e}_{\text{snap}}]_k \right] - \tan^{-1}(\beta/\alpha),
\end{align}
up to periodicity. By the polynomial DLA and slow Pauli expansion conditions (i.e. all DLA basis elements are expressed as linear combinations of Paulis), we can compute $\gamma_k$ in $\mathcal{O}(\text{poly}(n)\log(1/\epsilon))$ time.
\end{proof}

For illustrative purposes, we show in Figure \ref{fig:prod-map-easy} the snapshot inversion process for the special case where $i\mathsf{Z}_j \in \g$, i.e. 
\begin{align}
    x_j = \cos^{-1}\left(2\mathbf{\gamma}^{(j)}\cdot \mathbf{e}_{\text{snap}}\right),
\end{align} for $i\mathsf{Z}_j = \sum_{k=1}^{\dim(\g)}\gamma_k^{(j)}\mathbf{B}_k$. The general parallel data reuploading case can be handled by applying the procedure to only one of the rotations that encodes at $x_j$ at a time, checking to find one that does not cause the algorithm to return FAILURE.

\begin{figure}
    \centering
    \includegraphics[width=0.5\textwidth]{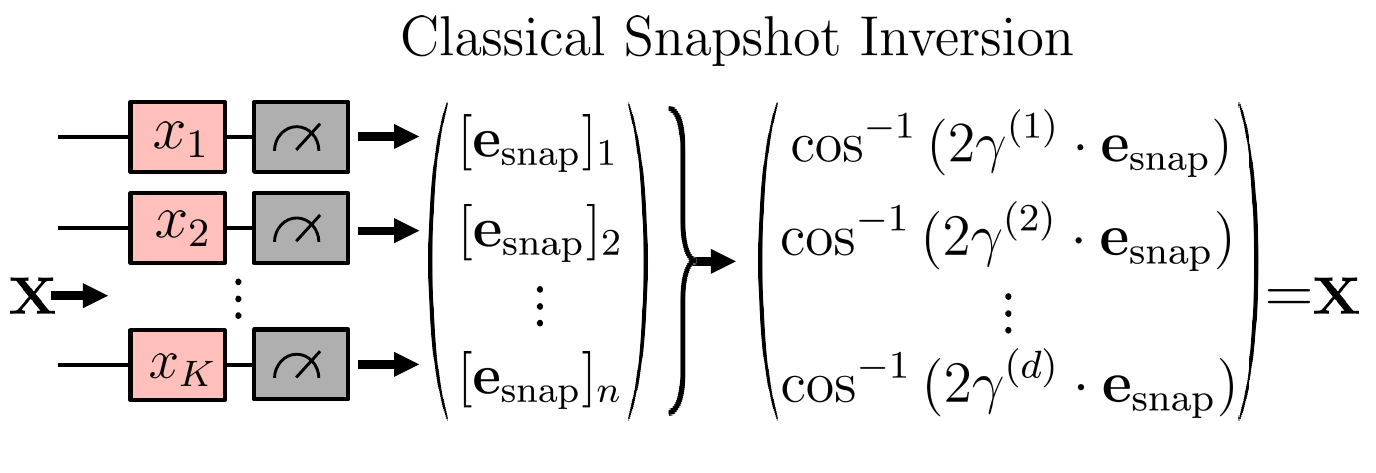}
    \caption{A product map encoding, whereby each input variable $x_j$ is encoded into an individual qubit, and the snapshot used by the model corresponds to single qubit measurements of the DLA basis elements. In this setting, the snapshot is trivial to invert and find the original data using the relation $x_j = \cos^{-1}\left(2\mathbf{\gamma}^{(j)}\cdot \mathbf{e}_{\text{snap}}\right) $.} 
    \label{fig:prod-map-easy}
\end{figure}

\subsubsection{General Pauli Encoding}\label{sec:general-pauli}

We now present a more general procedure that applies to feature maps that use serial data reuploading and multi-qubit Paulis. However, we introduce a condition that ensures that each $x_j$ is locally encoded. More generally, we focus our discussion on encoding states that may be written as a tensor product of $\Omega$ subsystems, i.e. multipartite states. 
\begin{align}
\label{eqn:separable_state}
\rho(\mathbf{x}) = \bigotimes_{ J \in \mathcal{P} } {\rho}_J(\mathbf{x}),
\end{align}
where $\dim(\mathsf{x}_J)$ is constant. The procedure is highlighted in Algorithm \ref{alg:snapshot_inver_gen_pauli} and requires solving a system of polynomial equations. 

In addition, the procedure may not be completely classical as quantum assistance may be required to compute certain expectation values of $\rho_J(\mathbf{x})$, specifically with respect to the DLA basis elements. For simplicity, the algorithm and the theorem characterizing the runtime ignore potential errors in estimating these expectations. If classical estimation is possible, then we can potentially achieve a $\mathcal{O}(\text{poly}(\log(1/\epsilon)))$ scaling. However, if we must use quantum, then we will incur a $\mathcal{O}(1/\epsilon)$ (due to amplitude estimation) dependence, which can be significant. Theorem~\ref{thm:seperable-proof} presents the attack complexity ignoring these errors.

\begin{algorithm}
\begin{algorithmic}
\Require Snapshot vector $\mathbf{e}_{\text{snap}}(\mathbf{x})$ of dimension $\dim(\g) = \mathcal{O}(\text{poly}(n))$ corresponding to a basis $(\mathbf{B}_{k})_{k=1}^{\dim(\g))}$ of DLA $\g$. Each $\mathbf{B}_{k}$ is expressed as a linear combination of $\mathcal{O}(\text{poly}(n))$ Pauli strings. Snapshot inversion is being performed for a VQC model that utilized a trainable portion of $\mathbf{U}(\mathbf{\theta})$ with DLA $\g$ and separable encoding Equation \eqref{eqn:separable_state} with qubit partition $\mathcal{P}$. Index $j \in [d]$, $\epsilon < 1$
\Ensure  An $\epsilon$ estimate of the $j$-th component $x_j$ of the data input $\mathbf{x} \in \mathbb{R}^{d}$ up to periodicity

\begin{enumerate}
    \item Find a $\rho_J$ for $J \in \mathcal{P}$ that depends on $x_j$. Let $R$ denote the number of Pauli rotations in the circuit for preparing $\rho_J$ that involve $x_j$.
    \item For each $k \in [\dim(g)]$, compute $\textup{Tr}(\mathbf{B}_k\rho_J(\mathbf{x}))$ and $\textup{Tr}(\mathbf{B}_k\rho_{J^{c}}(\mathbf{x}))$.
    \item Determine the set $\mathcal{S}_J = \{ k : \textup{Tr}(\mathbf{B}_k\rho_J(\mathbf{x})) \neq 0~\&~ \textup{Tr}(\mathbf{B}_k\rho_{J^{c}}(\mathbf{x})) =0 \}$, $J^c := [n] - J$.
    \end{enumerate}
    \If{$\mathcal{S}_J < \dim(\mathbf{x}_J)$}
        \State \textbf{return} FAILURE
    \Else 
     \State {1. For each $k \in \mathcal{S}_J$ evaluate $\textup{Tr}(\mathbf{B}_k\rho_J(\mathbf{x}))$ at $M = 2R^{\dim(\mathsf{x}_J)}+1$ points, $\mathbf{x}_r \in \{\frac{2\pi r}{2R+1} : r = -R, \dots R\}^{\dim(\mathsf{x}_J)}$} 
     \State{2. For each $k$, solve linear system $$\text{Tr}(\mathbf{B}_k\rho_J(\mathbf{x}_{r})) = \alpha_0 + \sum_{r \in [R]^{\text{dim}(\mathsf{x}_J)}}\alpha_{r}e^{i\mathbf{r}\cdot \mathbf{x}_r}$$ \quad\quad for $\alpha$'s.}
    \State 3. Consider the polynomial system:
        \begin{align}
    &[\mathbf{e}_{\text{snap}}]_k = \Re\left[\alpha_0 + \sum_{r \in [R]^{\text{dim}(\mathsf{x}_J)}}\alpha_{\mathbf{r}}\prod_{j=1}^{\dim(\mathsf{x}_J)}(T_{r_j}(u_j) + iv_jU_{r_j-1}(u_j))\right],  k \in \mathcal{S}_J\\
    & u_j^2 + v_j^2 = 1, j \in J,
\end{align}
    \quad\quad where $u_j = \cos(x_j), v_k =\sin(x_j)$ and $T_{r}, U_{r}$ relate to Chebyshev polynomials.
    \State 4. Apply Buchberger's to obtain a Gr\"obner basis for the system.
    \State 5. Back substitution and univariable root-finding algorithm \cite{nocedal1999numerical} (e.g., Jenkins-Traub \cite{jenkins1970three}) to obtain $\tilde{\mathbf{x}}_J$.
    \State 6. \textbf{return} $\tilde{\mathbf{x}}_j$
    \EndIf
\end{algorithmic}
\caption{: Snapshot Inversion for General Pauli Encodings}
\label{alg:snapshot_inver_gen_pauli}
\end{algorithm}

\begin{theorem}\label{thm:seperable-proof}
    Suppose that the feature encoding state $\rho(\mathbf{x})$ is a multipartite state, specifically there exists a partition $\mathcal{P}$ of qubits $[n]$ such that
    $$\rho(\mathbf{x}) = \bigotimes_{ J \in \mathcal{P} } {\rho}_J(\mathbf{x}),$$
    where we define  $\mathsf{x}_J \subseteq \mathbf{x}$ to be components of $\mathbf{x}$ on which $\rho_{J}$ depends. In addition, we have as input an $\mathcal{O}(\textup{poly}(n))$-dimensional snapshot vector $\mathbf{e}_{\text{snap}}$ with respect to a known basis $\mathbf{B}_k$ for the DLA of the VQC.
    
    Suppose that for $\rho_{J}(\mathbf{x})$ the following conditions are satisfied:
    \begin{itemize}
        \item $\textup{dim}(\mathbf{x}_J)  =  \mathcal{O}(1)$,
        \item each $x_k$ is encoded at most $R=\mathcal{O}(\textup{poly}(n))$ times in, potentially multi-qubit, Pauli rotations.
        \item and the set $\mathcal{S}_J = \{ k : \text{Tr}(\mathbf{B}_k\rho_J(\mathbf{x})) \neq 0~\&~\text{Tr}(\mathbf{B}_k\rho_{J^c}(\mathbf{x})) =0 \}$ has cardinality at least $\textup{dim}(\mathsf{x}_J)$, where $J^c := [n] - J$.
    \end{itemize}
then the model admits quantum-assisted snapshot inversion for recovering $\mathsf{x}_J$. Furthermore, a classical snapshot inversion can be performed if $\forall k, \textup{Tr}(\mathbf{B}_k\rho_J(\mathbf{x}))$ can be evaluated classically for all $\mathbf{x}$. Overall, ignoring error in estimating $\textup{Tr}(\mathbf{B}_k\rho_J(\mathbf{x}))$, with the chosen parameters, this leads to a $\mathcal{O}(\textup{poly}(n, \log(1/\epsilon)))$ algorithm.
\end{theorem}

\begin{proof}

Given that each $x_k$ is encoded with multiqubit Pauli rotations, i.e. possible eigenvalues are $1$ and $-1$, it is well known \cite{Wierichs2022generalparameter} that the following holds:
\begin{align}
      f_k(\mathsf{x}_J) = \text{Tr}(\mathbf{B}_k\rho_J(\mathbf{x})) =  \alpha_0 + \sum_{r \in [R]^{\text{dim}(\mathsf{x}_J)}}\alpha_{r}e^{i\mathbf{r}\cdot \mathbf{x}_r}, \forall k \in \mathcal{S},
\end{align}
and $\text{Tr}(\mathbf{B}_k\rho_J(\mathbf{x}))$ is real. The set $\mathcal{S}_J$ is to ensure that we can isolate a subsystem where $\dim(\mathsf{x}_J)$ is constant. 
To ensure that the number of terms is $\mathcal{O}(\text{poly}(n))$  it suffices to restrict to $\text{dim}(\mathsf{x}_J)=\mathcal{O}(\log(n)), R= \mathcal{O}(\log(n))$. %
The $\alpha$ coefficients can be computed by evaluating $\text{Tr}(\mathbf{B}_k\rho_J(\mathbf{x}))$ at $2R^{\text{dim}(\mathsf{x}_J)} +1 = \mathcal{O}(\text{poly}(n))$ different points $\mathbf{x}'$. 
Depending on whether $\text{Tr}(\mathbf{B}_k\rho_J(\mathbf{x}))$ can be evaluated classically or quantumly implies whether this falls under classical or quantum-assisted snapshot inversion. This leads to a system of $\text{dim}(\mathsf{x}_J)$ equations in $\mathsf{x}_J$:
\begin{align}
    [\mathbf{e}_{\text{snap}}]_k = f_k(\mathsf{x}_J), k=1, \dots, \dim(\g).
\end{align}
Using the Chebyshev polynomials $T_n$, $U_n$ of the first and second kind, respectively, we can expressed the system as a system of polynomial equations with additional constraints:
\begin{align}
\label{eqn:system}
    &[\mathbf{e}_{\text{snap}}]_k = \Re\left[\alpha_0 + \sum_{r \in [R]^{\text{dim}(\mathsf{x}_J)}}\alpha_{\mathbf{r}}\prod_{j=1}^{\dim(\mathsf{x}_J)}(T_{r_j}(u_j) + iv_jU_{r_j-1}(u_j))\right],  k \in \mathcal{S}_J\\
    & u_j^2 + v_j^2 = 1, j \in J,
\end{align}
where $u_j = \cos(x_j), v_k =\sin(x_j)$. In addition, we use the Chebyshev polynomials defined as $\cos(n\theta) = T_n(\cos(\theta))$ and $\sin(\theta)U_{n-1}(\cos(\theta))=\sin(n\theta)$.
By our assumption that the DLA is polynomial, we have $\mathcal{O}(\text{poly}(n))$ equations in $2\dim(\mathsf{x}_J) = \mathcal{O}(\log\log(n))$ unknowns.

If all conditions until now are satisfied, we will have successfully written down a system of determined simultaneous equations. Considering bounds from computational geometry we note that in the worst-case of Buchberger's algorithm \cite{buchberger1985} the degrees of a reduced Gröbner basis are bounded by
\begin{equation}
M = 2 \Big( \frac{\Delta^2}{2} + \Delta \Big)^{2^{Q - 2}},
\end{equation}
where $\Delta$ is the maximum degree of any polynomial and $Q$ is the number of unknown variables \cite{dube1990grobner}. For a system of linear equations, it was shown that a worst case degree bound grows double exponentially in the number of variables \cite{MAYR1982305}. The maximum degree of any equation in Eq~\ref{eqn:system} is $\Delta = R^{\dim(\mathsf{x}_J)}$, and $Q = 2\dim(\mathsf{x}_J)$ so that
\begin{align}
    M = \mathcal{O}(R^{2\dim(\mathsf{x}_J)2^{\dim(\mathsf{x}_J)}}),
\end{align}
so for our chosen conditions the maximum degree is bounded by $M = \mathcal{O}(\text{poly}(n))$.

Buchberger's algorithm provides a Gr\"{o}bner basis in which backsubstituion could be used to solve equations in one variable. Numerical methods for solving polynomials in one variable generally scale polynomially in the degree. For solving each univariate polynomial at each step of the back substitution, we can apply a polynomial root-finding method, such that Jenkins--Traub \cite{jenkins1970three}, which can achieve at least quadratic global convergence (converge from any initial point and at a rate that is at least $\log\log(1/\epsilon)$). This leads to an overall $\mathcal{O}(\text{poly}(n, \log(1/\epsilon))$ algorithm, ignoring error in estimating $\Tr(\mathbf{B}_j\rho_J)$.
\end{proof}

In the case a circuit has an encoding structure that leads to a separable state, we have indicated conditions that guarantee snapshot inversion can be performed. If the model is also snapshot recoverable, by having a polynomially sized DLA, then this means the initial data input can be fully recovered from the gradients, and hence the attack constitutes a strong privacy breach. 

\subsection{Snapshot Inversion for Generic Encodings}\label{sec:approx-inversion}

In the general case but still $\dim(\g) = \mathcal{O}(\text{poly}(n))$, where it is unclear how to make efficient use of our knowledge of the circuit, we attempt to find an $\mathbf{x}$ via black-box optimization methods that produces the desired snapshot signature. More specifically, suppose for simplicity we restrict our search to $[-1, 1]^{d}$. We start with an initial guess for the input parameters, denoted as $\mathbf{x}'$, and use these to calculate expected snapshot values $\text{Tr}[\mathbf{B}_k\rho(\mathbf{x}')]$. A cost function can then be calculated that compares this to the true snapshot, denoted $\mathbf{e}_{\text{snap}}$. As an example, one can use the mean squared error as the cost function,
\begin{equation}
\label{eqn:cost}
    f(\mathbf{x}') = \| \mathbf{e}_{\text{snap}} - (\text{Tr}[\mathbf{B}_k\rho(\mathbf{x}')])_{k=1}^{\dim(\g)} \|^2_2 = \sum_{k \in [\text{dim}(\g)]} \left([\mathbf{e}_{\text{snap}}]_k -  \text{Tr}[\mathbf{B}_k\rho(\mathbf{x}')]\right)^2.
\end{equation}

The goal will be to solve the optimization problem $\min_{\mathbf{x}' \in [-1, 1]^{d}} f(\mathbf{x}')$. For general encoding maps, it appears that we need to treat this as a black-box optimization problem, where we evaluate the complexity in terms of the evaluations of $f$ or, potentially, its gradient. However, in our setting, it is unclear what is the significance of finding approximate local minimum, and thus it seems for privacy breakage, we must resort to an exhaustive grid search. For completeness, we still state results on first-order methods that can produce approximate local minima.

We start by reviewing some of the well-known results for black-box optimization. We recall Lipschitz continuity by,
\begin{definition}[$L$-Lipschitz Continuous Function] A function $f : \mathbb{R}^d \rightarrow \mathbb{R}$ is said be be $L$-Lipschitz continuous if there exists a real positive constant $L > 0$ for which,
$$\rvert f(\mathbf{x}) - f(\mathbf{y}) \rvert \leq L \rVert \mathbf{x} - \mathbf{y} \rVert_{2}.$$
\end{definition}

If we consider the quantum circuit as a black-box $L$-Lipschitz function and $\mathbf{x}'$ in some convex, compact set with diameter $P$ (e.g. $[-1, 1]^{d}$ with diameter $2\sqrt{d}$). One can roughly upper bound $L$ by the highest frequency component of multidimensional trig series for $f$, which can be an exponential in $n$ quantity.
In this case the amount of function evaluations that would be required to find $\mathbf{x}'$ such that $\| \mathbf{x} - \mathbf{x}'\|_2 \leq \epsilon$  would scale as \begin{align} \label{eqn:gridsearch}
\mathcal{O}\left( P\left(\frac{L}{\epsilon}\right)^d \right),
\end{align}
which is the complexity of grid search \cite{nesterov2018lectures}.
Thus if for constant $L$ this is a computationally daunting task, i.e. exponential in $d =\Theta(n)$.

As mentioned earlier, we it is possible to resort to first-order methods to obtain, an effectively dimension-independent, algorithm for finding an approximate local minima. We recall the definition of  $\beta$-smoothness as,
\begin{definition}[$\beta$-Smooth Function] 
A differentiable function $f : \mathbb{R}^d \rightarrow \mathbb{R}$ is said be be $\beta$-smooth if there exists a real positive constant $\beta > 0$ for which
$$ \rVert \nabla f(x) - \nabla f(y) \rVert_{2} \leq \beta \rVert x - y \rVert_2.$$
\end{definition}
If we have access to gradients of the cost function with respect to each parameter, then using perturbed gradient descent \cite{jin2017escape} would roughly require
\begin{equation}\label{eqn:var-scaling-lips}
     \tilde{\mathcal{O}}\Big(  \frac{P L\beta}{\epsilon^2} \Big),
\end{equation}
function and gradient evaluations for an $L$-Lipschitz function that is $\beta$-smooth to find an approximate local min. With regards to first-order optimization, computing the gradient of $f$ can be expressed in terms of computing certain expectations values of $\rho$, either via finite-difference approximation or the parameter-shift rule for certain gate sets \cite{Wierichs2022generalparameter}. 

Regardless whether recovering an approximate local min reveals any useful information about $\mathbf{x}$, up to periodicity, it is still possible to make such a task challenging for an adversary. In general, the encoding circuit will generate expectation values with trigonometric terms. To demonstrate, we can consider a univariate case of a single trigonomial $f(x) = \sin{(\omega x)}$, with frequency $\omega$. This function will be $\omega$-Lipschitz continuous with $\omega^2$-Lipchitz continuous gradient. Hence, when considering the scaling of gradient-based approach in Eq~\ref{eqn:var-scaling-lips} we see that the frequency of the trigonometric terms will directly impact the ability to find a solution. Hence, if selecting a frequency that scales exponentially $\omega = \mathcal{O}(\exp(n))$, then snapshot inversion appears to be exponentially difficult with this technique.

Importantly, if the feature map includes high frequency terms, for example the Fourier Tower map of \cite{kumar2023expressive}, then $\beta$ and $L$ can be $\mathcal{O}(\text{exp}(n))$. However, as noted in Section \ref{sec:separable_enc} it is possible to make use of the circuit structure to obtain more efficient attacks. In addition, a poor local minimum may not leak any information about $\mathbf{x}$.

\subsubsection{Direct Input Recovery}

Note that it also may be possible to completely skip the snapshot recovery procedure and instead variationally adjust $\mathbf{x}'$ so that the measured gradients of the quantum circuit $C_j'$, match the known gradients $C_j$ with respect to the actual input data. This approach requires consideration of the same scaling characteristics explained in Eq~\ref{eqn:var-scaling-lips}, particularly focusing on identifying the highest frequency component in the gradient spectrum. If the highest frequency term in the gradient $C_j$ scales exponentially, $\omega = \mathcal{O}(\text{exp}(n))$, then even gradient descent based methods are not expected to find an approximate local min in polynomial time.

Further privacy insights can be gained from Eq~\ref{eqn:grad-in-exp-snap} where a direct relationship between the gradients and the expectation value snapshot is shown, which in general can be written as
\begin{equation}
    C_j(\mathbf{x}) = \chi_t^{(j)}\cdot \mathbf{e}_{\text{snap}}(\mathbf{x}).
\end{equation}
This indicates that the highest frequency terms of any $ \mathbf{e}_{\text{snap}} $ component will also correspond to the highest frequency terms in $ C_j(\mathbf{x})$, as long as its respective coefficient is non-zero $\chi_t^{(j)} \neq 0$.

This underscores scenarios where direct input recovery may prove more challenging compared to snapshot inversion, particularly in a VQC model. Consider a subset $\tilde{\mathbf{e}}_{\text{snap}} \subseteq \mathbf{e}_{\text{snap}}$ where each component has the highest frequency that scales polynomially with $n$. If there are sufficiently many values in $\tilde{\mathbf{e}}_{\text{snap}} $ then recovering approximate local min to Eq~\ref{eqn:cost} may be feasible for these components. %
However, for gradient terms $C_j(\mathbf{x})$ that depend on all values of $\mathbf{e}_{\text{snap}}$, including terms outside of  $\tilde{\mathbf{e}}_{\text{snap}} $ that exhibit exponential frequency scaling, then gradient descent methods may take exponentially long when attempting direct input inversion, even if recovering approximate local minima to the  snapshot inversion task can be performed in polynomial time. 

Investigations into direct input recovery have been covered in previous work \cite{kumar2023expressive} where the findings concluded that the gradients generated by $ C_j(\mathbf{x})$ would form a loss landscape dependent on the highest frequency $\omega$ generated by the encoding circuit, indicating that exponentially scaling frequencies led to models that take exponential time to recover the input using quantum-assisted direct input recovery. The Fourier tower map encoding circuit used in \cite{kumar2023expressive} was designed such that $\omega$ scales exponentially to provide privacy, this was done by using $m$ qubits in a sub-register per data input $x_j$, with the single qubit rotation gates parameterized by an exponentially scaling amount. The encoding can be defined as
\begin{equation}
    \bigotimes_{j=1}^{d}\Big( \bigotimes_{l=1}^{m} R_{\mathsf{X}}( 5^{l-1} x_j) \Big).
\end{equation}
Hence, the gradient contained exponentially scaling highest frequency terms, leading to a model where gradient descent techniques took exponential time. However, if considering the expectation value of the first qubit in a sub-register of this model, we note this corresponds to a frequency $\omega = 1$ and hence the respective expectation value for the first qubit would be snapshot invertible. However, in the case of \cite{kumar2023expressive} the DLA was exponentially large, meaning the model was not snapshot recoverable, hence these snapshots could not be found to then be invertible. Hence, from our new insights, we can conclude that the privacy demonstrated in \cite{kumar2023expressive} was dependent on having an exponential DLA dimension. However, an exponentially large DLA also led to an untrainable model, limiting the real-world applicability of this previous work. Lastly, recall that Algorithm \ref{alg:snapshot_inver_pauli} in the case poly DLA and slow Pauli expansion is a completely classical snapshot inversion attack for the Fourier tower map. Further, highlighting how snapshot inversion can be easier than direct inversion.

We show that both direct input recovery and snapshot inversion are dependent on frequencies $\omega$ generated by the encoding circuit, highlighting that this is a key consideration when constructing VQC models.
The introduction of high frequency components can be used to slow down methods that obtain approximate local minimum to Eq~\ref{eqn:cost}. However, for true privacy breakage, it appears that in general we still need to resort to grid search, which becomes exponentially hard with dimension regardless of high-frequency terms. However, for problems with a small amount of input, introducing high frequency terms can be used to also make grid search harder. The idea of introducing large frequencies is a proxy for the more general condition that our results hint at \emph{for privacy}, which is that \emph{the feature map $\rho(\mathbf{x}')$ should be untrainable in terms of varying $\mathbf{x}'$}.

Notably, cases exist where the same model can have an exponential frequency gradient, but can still contain a certain number of expectation snapshot values with polynomial scaling frequencies. Hence, it is also important to note that merely showing that a model is not direct input recoverable does not guarantee privacy, as one needs to also consider that if the model is snapshot recoverable and that these snapshots may be invertible if sufficient polynomial scaling frequency terms can be recovered. This duality highlights the complexity of ensuring privacy in quantum computing models and stresses the need for a comprehensive analysis of the frequency spectrum in both model construction and evaluation of privacy safeguards.

\subsubsection{Expectation Value Landscape Numerical Results}

In this section, we provide a numerical investigation of the impact of high-frequency components in the encoding circuit on the landscape of Eq~\ref{eqn:cost} for snapshot inversion. The idea is to present examples that move beyond the Fourier tower map.
We present two cases of encodings that would generally be difficult to simulate classically. By plotting a given expectation value against a univariate $x$ we can numerically investigate the frequencies produced by both models. 

\begin{figure}[h!]
    \centering
    \includegraphics[width=0.5\textwidth]{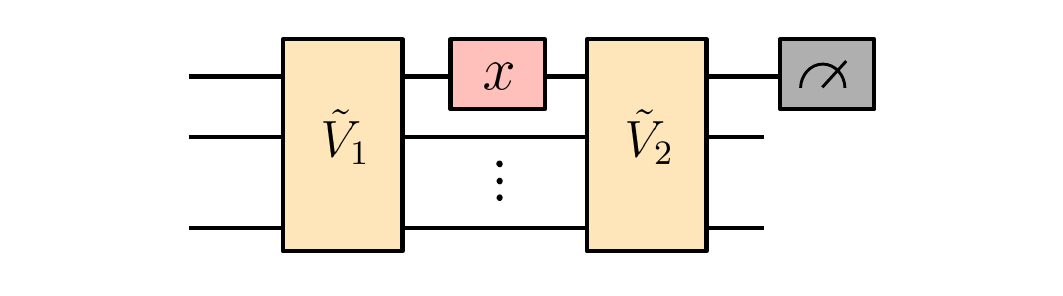}
    \caption{Encoding circuit diagram showing a single qubit $R_{\mathsf{X}}$ rotation gate parameterised by the univariate parameter $x$, but with arbitrary $2^n$ dimensional unitaries applied before and after the $x$ parameterized gate. Despite being hard to simulate analytically, the expectation value $e_{\text{in}}$ varies as a simple sinusoidal function in $x$, regardless of the total number of qubits $n$.}
    \label{fig:single-qsim-scaling}
\end{figure}

In Figure~\ref{fig:single-qsim-scaling} we demonstrate a circuit in which $x$ parameterizes a single $R_{\mathsf{X}}$ rotation gate, but on either side of this is an unknown arbitrary unitary matrix acting on $n$ qubits. This would be classically hard to simulate due to the arbitrary unitary matrices; however, the result effectively corresponds to taking measurements on an unknown basis, and using only a few samples of $x$ it is possible to recreate the graph as a single frequency sinusoidal relationship. This results in the distance between the stationary points being $r=\pi$ for any value of $n$. This corresponds to a frequency $\omega = \frac{r}{\pi} = 1$, regardless of the value of $n$. This circuit therefore exhibits constant frequency scaling independent of $n$ and hence could be easy for gradient-based methods to recovery approximate local min. 

\begin{figure}[h!]
    \centering
    \includegraphics[width=0.4\textwidth]{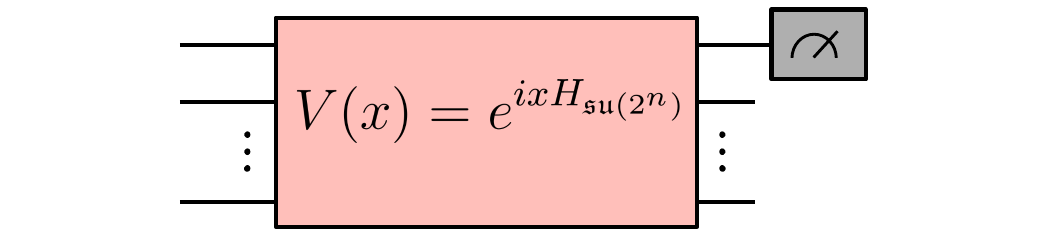}
    \caption{Encoding circuit diagram showing a $\text{SU}(2^n)$ gate parameterised by a univariate parameter $x$. }
    \label{fig:sun-qsim-scaling}
\end{figure}

\begin{figure}[h!]
    \centering
    \begin{subfigure}{0.4\textwidth}
        \includegraphics[width=\textwidth]{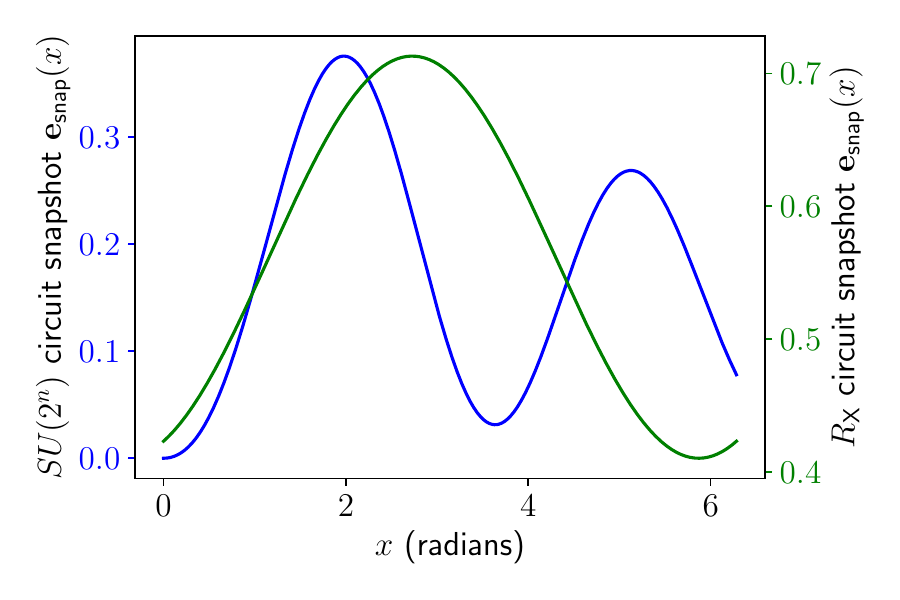}
        \caption{Number of qubits $n = 2$.}
    \end{subfigure}
    \begin{subfigure}{0.4\textwidth}
        \includegraphics[width=\textwidth]{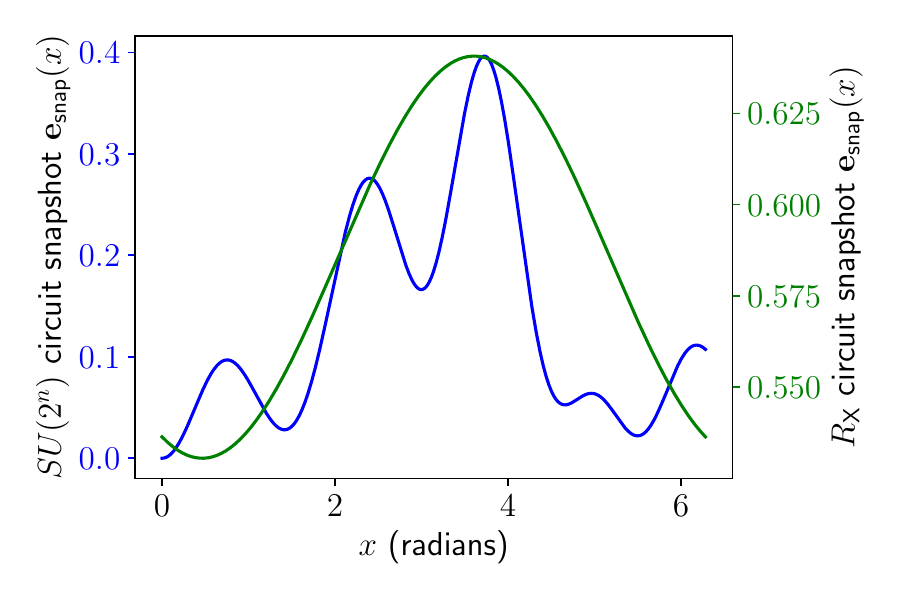}
        \caption{Number of qubits $n = 3$.}
    \end{subfigure}
    \begin{subfigure}{0.4\textwidth}
        \includegraphics[width=\textwidth]{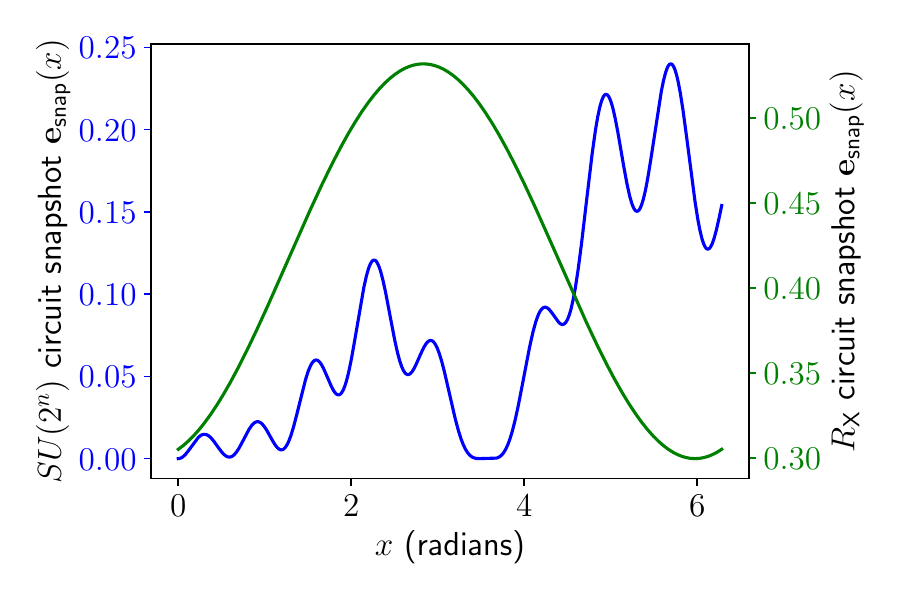}
        \caption{Number of qubits $n = 4$.}
    \end{subfigure}   
    \begin{subfigure}{0.4\textwidth}
        \includegraphics[width=\textwidth]{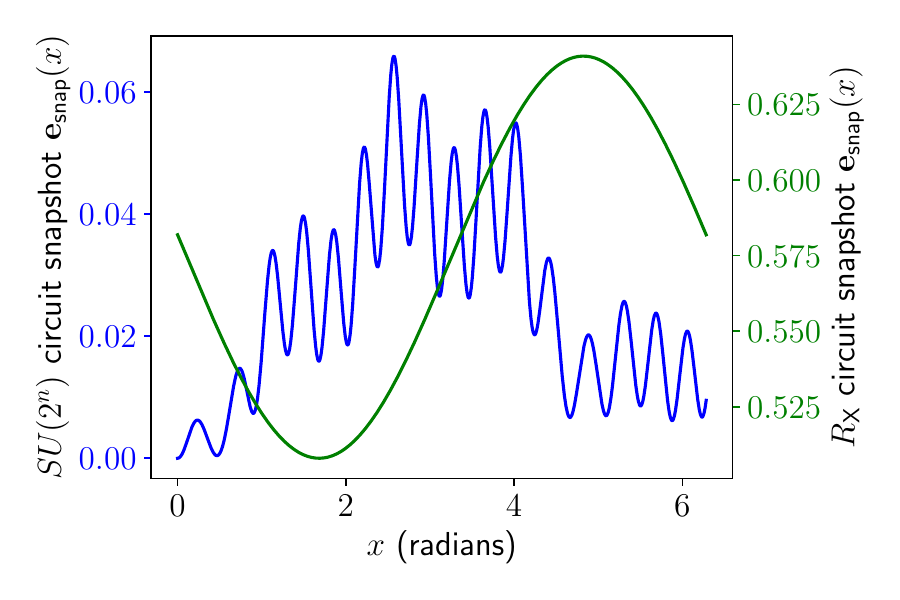}
        \caption{Number of qubits $n = 5$.}
    \end{subfigure}   
    \caption{Comparison of how expectation value of the measurement of $Z_1$ varies with $x$ for both the model parameterized using a single $R_{\mathsf{X}}$ rotation gate as detailed in Figure~\ref{fig:single-qsim-scaling} and the model parameterized using an $SU(2^n)$ gate as detailed in Figure~\ref{fig:sun-qsim-scaling}, for varying amounts of qubits.}
    \label{fig:expectation-comparison}
\end{figure}

\begin{figure}[h!]
    \centering
    \includegraphics[scale=0.65]{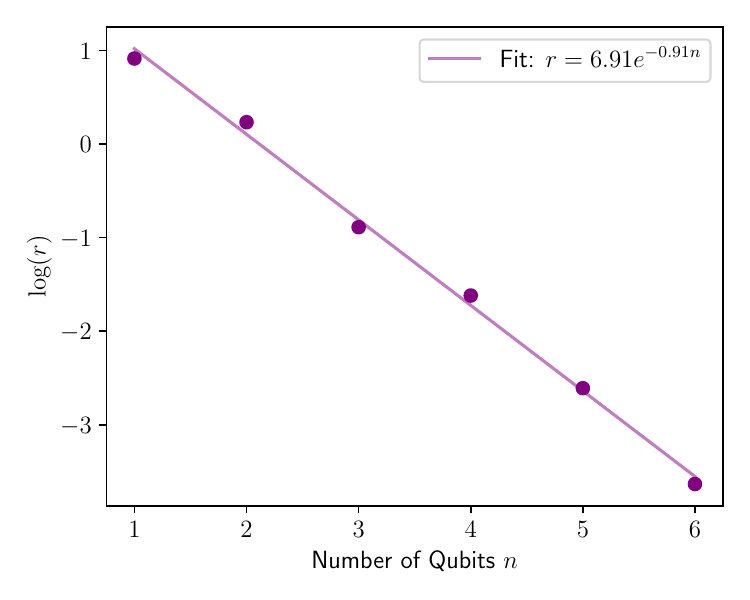}
    \caption{Plot showing the relationship between the average minimum distance $r$ between stationary points of the expectation value $\mathsf{Z}\otimes\mathbb{I}^{\otimes n-1}$ as a function of a univariate $x$ input. The encoding circuit considered is a parameterized $\text{SU}(2^n)$ gate which is parameterized by a univariate input $x$ as $U = e^{i H x}$, where $H$ is a randomly generated Hermitian matrix. The average was taken over ten repeated experiments where $H$ was regenerated each time.}
    \label{fig:iqp-scaling}
\end{figure}

We briefly give an example of a type of circuit that can generate high-frequency expectation values. Figure~\ref{fig:sun-qsim-scaling} demonstrates a circuit where $x$ parameterizes an $\text{SU}(2^n)$ gate. The result when measuring the same expectation value corresponds to a highest frequency term that is exponentially increasing. This is shown in the plot in Figure~\ref{fig:sun-qsim-scaling} in which the distance between stationary points $r$ shrinks exponentially as the number of qubits increases for the $\text{SU}(2^n)$ parameterized model, which roughly corresponds to an exponentially increasing highest frequency term. A comparison between the expectation value landscape of the two different encoding architectures, is shown in Figure~\ref{fig:expectation-comparison}, demonstrating that the single rotation gate parameterization, as shown in Figure~\ref{fig:single-qsim-scaling}, produces a sinusoidal single-frequency distribution, even as the number of qubits is increased; while the $\text{SU}(2^n)$ gate parameterization, shown in Figure~\ref{fig:sun-qsim-scaling}, contains exponentially increasing frequency terms. A visual representation for the multivariate case is also demonstrated in Figure~\ref{fig:expectation-comparison-3d} which shows the expectation value landscape when two input parameters are adjusted, for a model comprised of two different $\text{SU}(2^n)$ parameterized gates parameterized by the variables $x_1$ and $x_2$ respectively, demonstrating that as more qubits are used, the frequencies of the model increase and hence so does the difficulty of finding a solution using gradient descent techniques.

The two example circuits demonstrate encoding circuits that are hard to simulate, and hence no analytical expression for the expectation values can be easily found. These models do not admit classical snapshot inversion; however, by sampling expectation values it may be possible to variationally perform quantum-assisted snapshot inversion. Whether numerical snapshot inversion can be performed efficiently will likely be affected by the highest frequency $\omega$ inherent in the encoding, which will itself depend on the architecture of the encoding circuit. This suggests that designing encoding circuits such that they contain high-frequency components is beneficial in high-privacy designs. We have shown that $\text{SU}(2^n)$ parameterized gates can produce high-frequency terms, whereas single-qubit encoding gates will be severely limited in the frequencies they produce.

\begin{figure}[h!]
    \centering
    \begin{subfigure}{0.3\textwidth}
        \includegraphics[width=\textwidth]{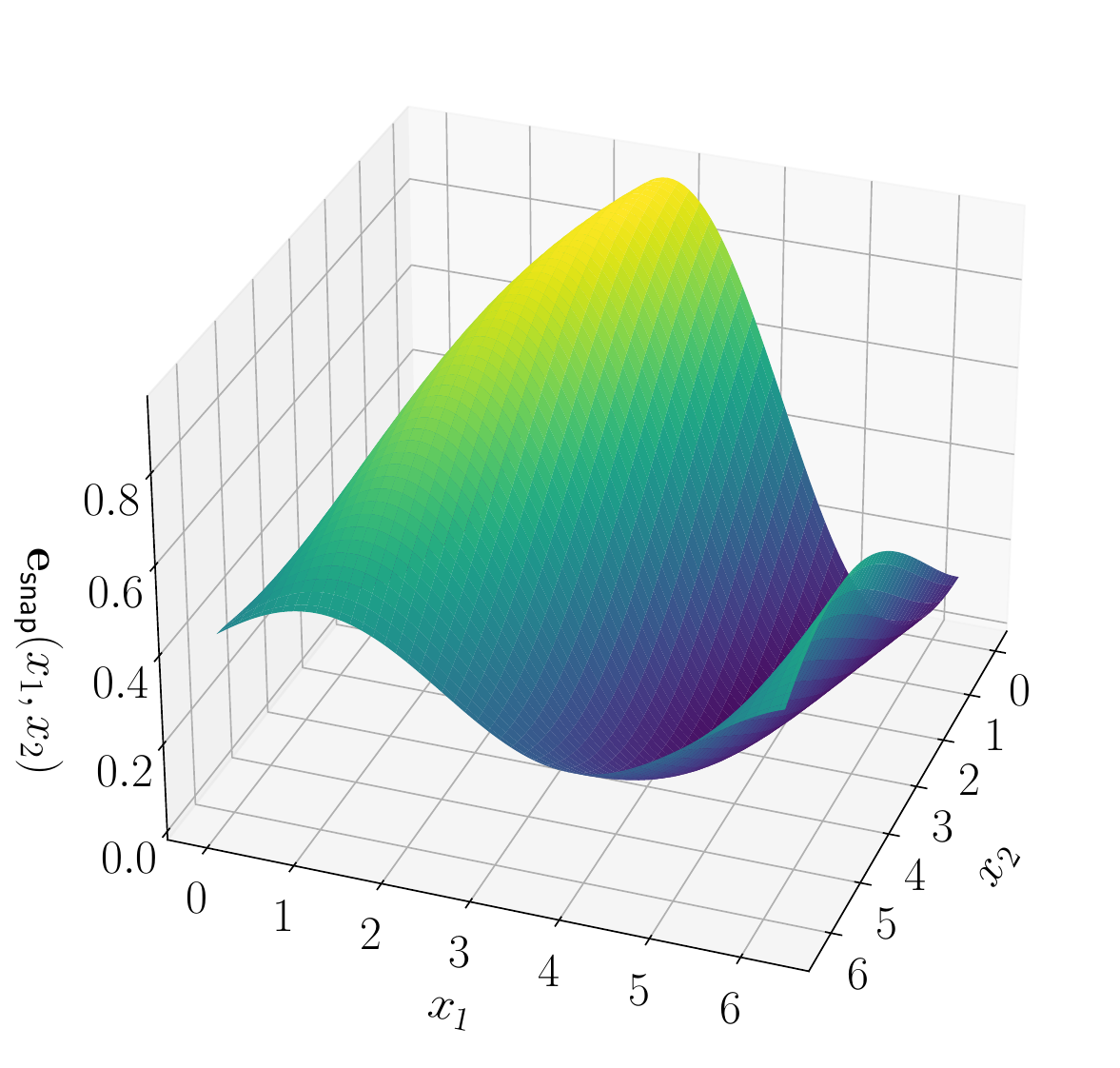}
        \caption{Number of qubits $n = 1$.}
    \end{subfigure}
       \begin{subfigure}{0.3\textwidth}
        \includegraphics[width=\textwidth]{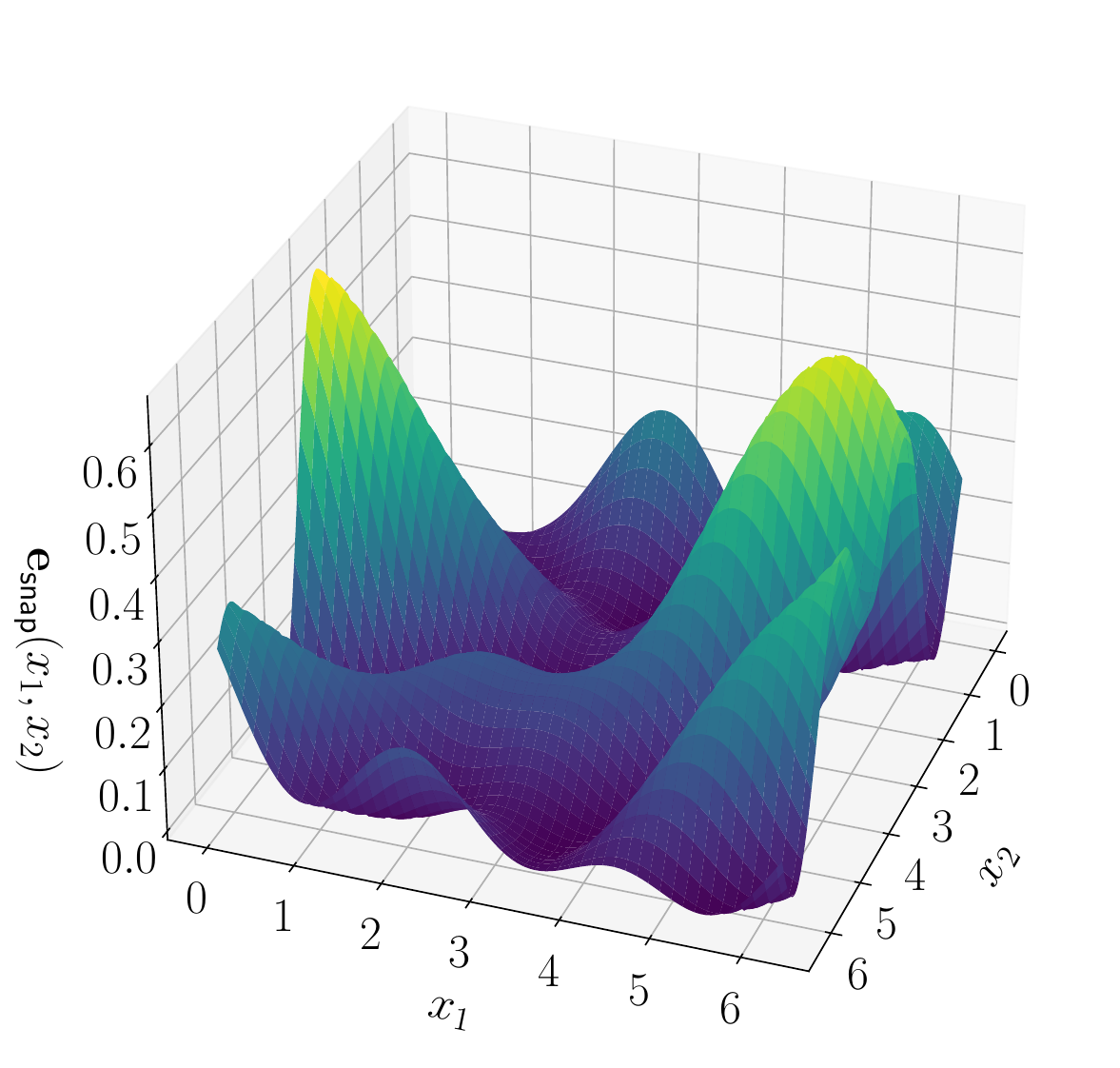}
        \caption{Number of qubits $n = 2$.}
    \end{subfigure}
    \begin{subfigure}{0.3\textwidth}
        \includegraphics[width=\textwidth]{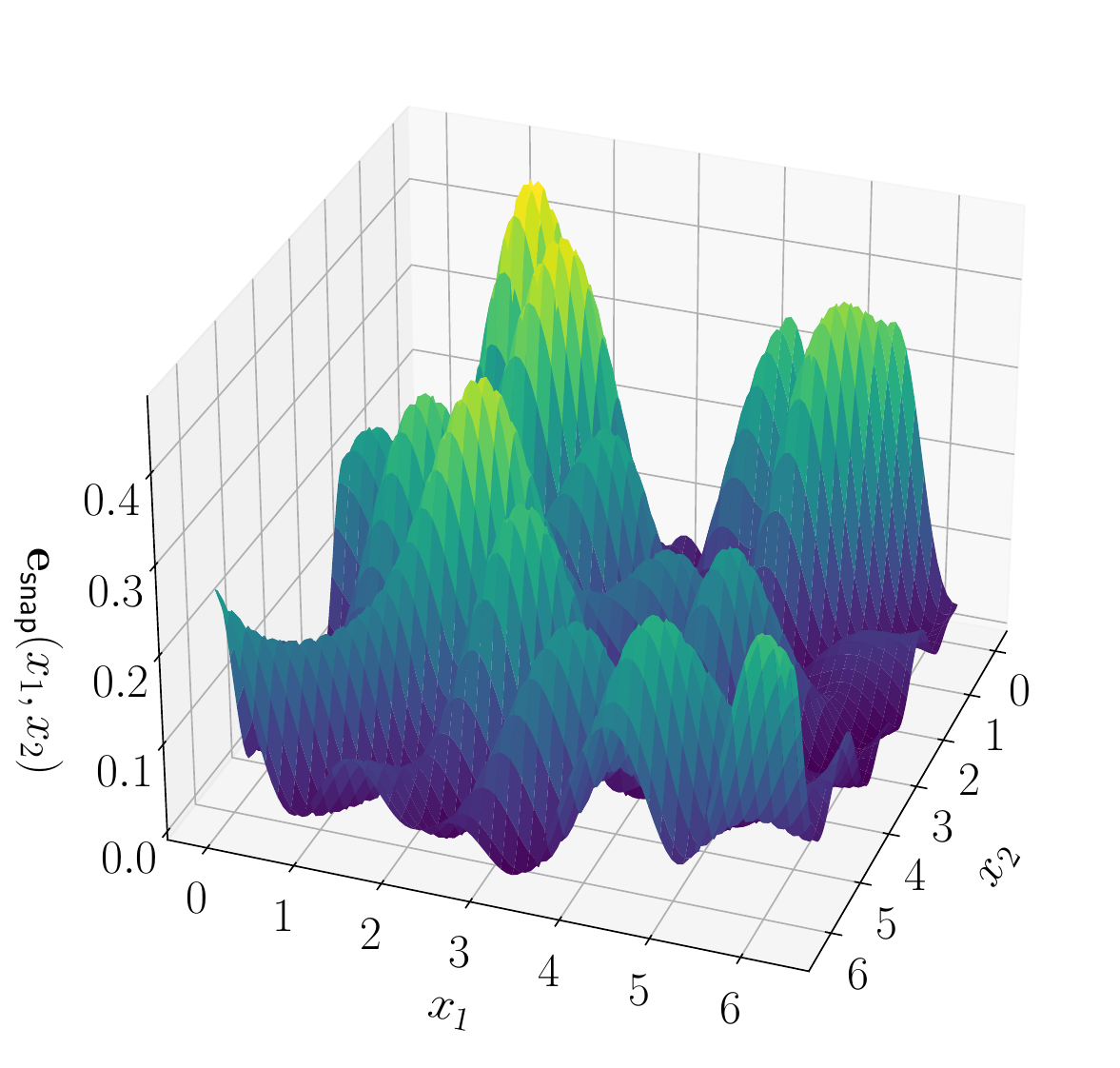}
        \caption{Number of qubits $n = 3$.}
    \end{subfigure}
    \caption{Comparison how expectation value of the measurement of $Z_1$ varies with $x_1$ and $x_2$ for a model parameterized by the circuit $v(x) = e^{i H_2 x_2} e^{i H_1 x_1}$, where $H_1$ and $H_2$ are randomly chosen Hermitian matrices of dimension $2^n$, where $n$ is the number of qubits in the circuit.}
    \label{fig:expectation-comparison-3d}
\end{figure}
\section{Discussion}

We utilize this section to draw the connections between the two key properties of VQC: \emph{trainability}, i.e., the lack of barren plateaus \cite{fontana2023adjoint, Ragone:2023qbn}, and the ability to retain privacy of input. Building upon this connection, we discuss the prospects and future of achieving robust privacy guarantees with VQC models. 

\subsection{Connections between Trainability and Privacy in VQC} \label{sec:train_privacy}

Solely requiring a machine learning model to be private is not sufficient to deploy it for a practical use case of distributed learning such as federated learning. A key requirement in this collaborative learning scenario is also to ensure that the model remains trainable. A plethora of works have gone into exactly characterizing the trainability of VQC models by analyzing the presence of barren plateau in the VQC model, starting from the work of \cite{mcclean2018barren} and culminating in the works of \cite{fontana2023adjoint, Ragone:2023qbn}. Especially, the work of \cite{fontana2023adjoint} provides an exact expression of the variance of the gradient of the model when the VQC is constrained to be in the LASA case, the details of which we also provide in Appendix~\ref{app:trainability} for completeness. A key insight into these works suggests that LASA models, with exponentially-sized DLA, may lead to the presence of barren plateaus (see Theorem \ref{thm:gradient thm} under Appendix~\ref{app:trainability}), drastically deteriorating the trainability of such models. 

Within our privacy framework centered around snapshot recoverability (Sec\ref{sec:sanpshot_recovery}), we also show via Theorem~\ref{thm:snaprecoveryscaling} that LASA models with an exponential size DLA are not classically snapshot recoverable, although this may lead to untrainable models. We can therefore conclude that a possible condition for protection against classical input recovery using gradients in a VQC model is to choose an ansatz that exhibits an exponentially large dynamical Lie algebra dimension, as this would render snapshot recovery difficult. Through our framework, we can see that previous works \cite{kumar2023expressive} effectively relied on this property to ensure privacy. Combining the concept of trainability leads to the following corollary on the privacy of VQC models:

\begin{corollary}\label{col:trainability-privacy}
Any trainable VQC on $n$ qubits that satisfies the LASA condition in Def~\ref{def:lasa}, fulfills the slow Pauli expansion condition as highlighted in Def~\ref{def:slow_pe}, and has a DLA $\g$ whose dimension scales as $\mathcal{O}(\text{poly}(n))$, would admit snapshot recoverability with complexity $\mathcal{O}(\text{poly}(n))$.

\end{corollary}
Hence we can conclude that, at least in the LASA case of VQC, the privacy of the model is linked to the DLA dimension and furthermore that there is a direct tradeoff between privacy and trainability of the model. As exponentially sized DLA models are expected to be untrainable in the LASA case, this means that for realistic applications, it does not seem feasible to rely on quantum privacy derived from an exponential DLA precluding snapshot recoverability in the model. This suggests that any privacy enhancement from quantum VQCs should not derive itself from the variational part of the circuit for LASA type models that are intended to be trainable. In other words, we expect the majority of trainable VQC models to be vulnerable to weak privacy breaches.  The privacy of variational models beyond the LASA case becomes linked to a larger question within the field, notably, whether there exist quantum variational models that are not classically simulatable and do not have barren plateaus \cite{cerezo2023does}.

It is also worth noting that if one attempted to create a model that is not snapshot recoverable by ensuring that  $D < {\text{dim}(\mathfrak{g})}$, and hence an underparameterized system of equations, it would effectively lead to an underparameterized model. A model is underparmeterized when there are not enough variational parameters to fully explore the space generated by the DLA of the ansatz, which is a property that may not be desirable for machine learning models \cite{Avron2020}.

\subsection{Future Direction of VQC Quantum Privacy} \label{sec:future_dir}

Due to the above argument suggesting that achieving privacy via an exponentially large DLA may cause trainability issues in the underlying model, it appears that future improvements in privacy using VQC may primarily focus on preventing the snapshot inversion step, as we highlight in Sec~\ref{sec:input_recover}. This promotes a focus on the encoding circuit architectures of the VQC in order to prevent the model admitting snapshot inversion to facilitate input recovery.

We have explicitly shown the necessary condition required to achieve privacy from purely classical attacks. If it is not possible to classically simulate the expectation values of the quantum encoded state with respect to the DLA basis elements of the variational circuit, then it will not be possible to attempt classical analytical or numerical inversion attacks. Any VQC designed where these expectation values cannot be simulated will, therefore, be protected against any purely classical snapshot inversion attempts. This condition can therefore prevent strong privacy breaches, as long as the attacking agent only has access to a classical device.

In the case where the attacker can simulate expectation values of the DLA basis or has access to a quantum device to obtain the expectation values, then numerical classical snapshot inversion or numerical quantum-assisted snapshot inversion can be attempted, respectively. We have shown that in this case an important factor in preventing these techniques is that the expectation values have exponentially scaling frequency terms resulting in the attacks requiring to solve a system of high-degree polynomial equations. The implication of this is that to achieve useful privacy advantage in VQC  it may require that the encoding circuit is constructed in such a way that the expectation values of the DLA basis elements of the variational circuit contain frequency terms that scale exponentially. Notably, we find that having high frequency terms in the gradients, as suggested in the encoding circuit of \cite{kumar2023expressive}, does not necessarily protect against numerical snapshot inversion attacks. This is because the gradients inherit the highest frequency term from all expectation values, but there may be a sufficient number of polynomial frequency expectation values to perform snapshot inversion, even if direct input inversion is not possible.

Unlike the variational case, where a connection between DLA dimension and trainability has been established, the effect that privacy-enhancing quantum encodings would have on the trainability of a model is less clear. If the majority of expectation values used in model contain exponentially large frequencies, then this potentially restricts the model to certain datasets. In classical machine learning, there have been positive results using trigonometric feature maps to classify high frequency data in low dimensions \cite{tancik2020fourier}. It remains as a question for future research the types of data which may be trained appropriately using the privacy-preserving high frequency feature maps proposed. If models of this form are indeed limited in number, then the prospects for achieving input privacy from VQC models appear to be limited.  More generally the prospect for quantum privacy rests on feature maps that are \emph{untrainable} with regard to adjusting $\mathbf{x}'$ to recover expectation values $\mathbf{e}_{\text{snap}}$, while at the same time remaining useful feature maps with respect to the underlying dataset and overall model.

\section{Conclusion} \label{sec:conclusion}

In this research, we conducted a detailed exploration of the privacy safeguards inherent in VQC models regarding the recovery of original input data from observed gradient information. Our primary objective was to develop a systematic framework capable of assessing the vulnerability of these quantum models to general class of inversion attacks, specifically through introducing the snapshot recovery and snapshot inversion attack techniques, which primarily depend on the variational and encoding architectures, respectively.

Our analysis began by establishing the feasibility of recovering snapshot expectation values from the model gradients under the LASA assumption. We demonstrated that such recovery is viable when the Lie algebra dimension of the variational circuit exhibits polynomial scaling in the number of qubits. This result underscores the importance of algebraic structure in determining the potential for privacy breaches in quantum computational models. Furthermore, due to the fact that a polynomial scaling DLA dimension is commonly required for models to be trainable, our results suggest that a trade-off may exist between privacy and trainability of VQC models. Assuming one insists on a polynomial-sized DLA, our framework suggests that a weak privacy breach will always be possible for the type of VQC model studied. To ensure privacy of the model overall, one cannot rely on the variational circuit and needs to instead focus more on the encoding architecture and ensuring snapshot inversion cannot be performed. If snapshot inversion is not possible, then at least strong privacy breaches can be prevented.

We then explored snapshot inversion, where the task is to find the original input from the snapshot expectation values, effectively inverting the encoding procedure. Studying widely used encoding ansatz, such as the local multiqubit Pauli encoding, we found that under the conditions that a fixed subset of the data paramaterizes a constituent state which has sufficient overlap with the DLA, and the number of gates used to encode each dimension of the input $\mathbf{x}$ was polynomial, that snapshot inversion was possible in $\mathcal{O}(\text{poly}(n, \log(1/\epsilon))$ time. This shows that a potentially wide range of encoding circuits are vulnerable to strong privacy breaches and brings their usage in privacy-focused models into question. For the most general encoding, which we approached as a black-box optimization problem, we demonstrated that using perturbed gradient descent to find a solution is constrained by the frequency terms within the expectation value Fourier spectrum. In general for exactly finding $\mathbf{x}$ it appears that a grid search would be required. Although we cannot provide strict sufficient conditions due to the possibility of unfavorable local minima with perturbed gradient, we note that gradient descent for snapshot inversion may, in some cases, be easier to perform for snapshot inversion than for direct input data recovery from the gradients. This simplification arises because gradients can inherit the highest frequency term from the snapshots, potentially leading to scenarios where the gradient term contains exponentially large frequencies. However, there may still be sufficient polynomial frequency snapshots to permit snapshot inversion. This shifts the focus in attack models away from direct input recovery from gradients, a common approach in classical privacy analysis, towards performing snapshot inversion as detailed in this study as a potentially more efficient attack method.

The dual investigation allowed us to construct a robust evaluative framework that not only facilitates the assessment of existing VQC models for privacy vulnerabilities but also aids in the conceptualization and development of new models where privacy is a critical concern. Our reevaluation of previous studies, such as those cited in \cite{kumar2023expressive}, through the lens of our new framework, reveals that the privacy mechanisms employed, namely the utilization of high-frequency components and exponentially large DLA, effectively prevent input data recovery via a lack of snapshot recoverability, but at the same time contribute to an untrainable model of limited practical use.

In conclusion, we offer a methodological approach for classifying and analyzing the privacy features of VQC models, presenting conditions for weak and strong privacy breaches for a broad spectrum of possible VQC architectures. Our findings not only enhance the understanding of quantum privacy mechanisms but also offer strategic guidelines for the design of quantum circuits that prioritize security while at the same time maintaining trainability. Looking ahead, this research paves the way for more robust quantum machine learning model designs, where privacy and functionality are balanced. This knowledge offers the potential to deliver effective machine learning models that simultaneously demonstrate a privacy advantage over conventional classical methods.

\section*{Acknowledgments}
The authors thank Brandon Augustino, Raymond Rudy Putra and the rest of our colleagues at the Global Technology Applied Research Center of JPMorgan Chase for helpful comments and discussions.

\section*{Disclaimer}

This paper was prepared for informational purposes by the Global Technology Applied Research center of JPMorgan Chase $\&$ Co. This paper is not a product of the Research Department of JPMorgan Chase $\&$ Co. or its affiliates. Neither JPMorgan Chase $\&$ Co. nor any of its affiliates makes any explicit or implied representation or warranty and none of them accept any liability in connection with this paper, including, without limitation, with respect to the completeness, accuracy, or reliability of the information contained herein and the potential legal, compliance, tax, or accounting effects thereof. This document is not intended as investment research or investment advice, or as a recommendation, offer, or solicitation for the purchase or sale of any security, financial instrument, financial product or service, or to be used in any way for evaluating the merits of participating in any transaction.

\bibliographystyle{ieeetr}
\bibliography{references}

\clearpage
\onecolumngrid
\appendix

\section{Notation used in the work} \label{sec:notations}

\begin{table}[H]
\centering
\begin{tabular}{cl}
\hline
\textbf{Symbol} & \textbf{Meaning} \\
\hline
 $\mathbf{x}$ & Data Input \\
 $\mathbf{\theta}$ & Variational Parameters \\
 $y_{\mathbf{\theta}}(\mathbf{x})$ & Model Output\\
 $C_j$ & Gradient of Model Output wrt $\theta_j$ \\
$n$ & Total number of qubits \\
$d$ & Dimension of input vector\\
$D$ & Number of variational parameters \\
$V(\mathbf{x})$ & Encoding Circuit \\
$\mathbf{U}(\mathbf{\theta}) $ & Variational Circuit \\
$O$ & Measurement Operator in VQC Model \\
$\mu_\alpha$ & Coefficients of $O$ when written in terms of supported DLA basis terms\\
$\mathfrak{g}$ & Dynamical Lie algebra \\
$\Omega$ & Maximum number of sub-register partitions\\
$\mathbf{B}_k$ & Basis of DLA \\
$B_k$ & Portion of basis DLA basis element that only acts on a given sub-register \\
$H_k$ & Hermitian gate generators of the encoding map circuit \\
$m$ & Number of qubits in sub-register \\
$\rho_J(\mathsf{x}_J)$& Constituent state in a separable state\\
$\Omega$ & Maximum number of constituent states in a separable state.\\
$A_k$ & Element of constituent separable state density matrix sum expansion\\
$\lambda$ & Terms when expanding constituent density matrix $\rho_J(\mathsf{x}_J)$\\
$P$ & Diameter of sphere of possible $\mathbf{x}$ values in $\mathbb{R}^d$ \\ 
$\epsilon$ & User defined acceptable error\\
\hline
\end{tabular}
\caption{Table of Notation}
\label{tab:notation}
\end{table}

\section{General Cost Function}\label{app:loss-function}

In this section we go into more detail regarding how the choice of cost function would affect the attack procedure. In practical real-world examples, we would most likely have gradients with respect to some cost function $\frac{\partial\texttt{Cost}(y_\theta, y)}{\partial \theta}$ rather than gradients in the form $C_j = \frac{\partial y_\theta}{\partial \theta}$ that are used in this paper (which would correspond to a linear cost function of the form $\texttt{Cost}(y_\theta, y) = y_\theta - g(y)$, where $g(y)$ is any function independent of $y_\theta$). We briefly show here that with a slight modification our results hold for any differentiable cost function.

Using the chain rule we can see that the $\theta_j$ gradient can be written
\begin{equation}
    \tilde{C}_j = \frac{\partial\texttt{Cost}(y_\theta, y)}{\partial \theta_j} = \frac{\partial\texttt{Cost}(y_\theta, y)}{\partial y_\theta} \cdot \frac{\partial y_\theta}{\partial \theta_j}.
\end{equation}
If the value of $\frac{\partial\texttt{Cost}(y_\theta, y)}{\partial y_\theta}$ can be calculated, which would be possible if we were given the output of the model $y_\theta$, along with the data labels $y$, then it would be possible to directly convert between $\tilde{C}_j$ and $C_j$. When this is not the case, it is still possible to attempt a solution, although it requires one additional gradient than usual. We show this by considering that the $\frac{\partial\texttt{Cost}(y_\theta, y)}{\partial y_\theta}$ term, although unknown, will be the same for all $\theta_j$. Hence we can eliminate this term considering ratios of the known gradients defined as 

\begin{equation}
    R_j \equiv \frac{\tilde{C}_j}{\tilde{C}_1} = \frac{\frac{\partial y_\theta}{\partial \theta_j}}{\frac{\partial y_\theta}{\partial \theta_1}}.
\end{equation}

This allows us to write new equations for all $j >1$ without this unknown term as
\begin{equation}
    \frac{\partial y_\theta}{\partial \theta_j}= R_j \cdot \frac{\partial y_\theta}{\partial \theta_1},
\end{equation}
hence in the notation used throughout this work 
\begin{equation}
    C_j = R_j \cdot C_1.
\end{equation}
In Eq~\ref{eqn:grad-in-exp-snap} it was previously shown that the $C_j$ gradients can be written in terms of snapshot components as
\begin{equation}\label{eqn:grad-in-exp-snap-2}
    C_j = \sum_{t=1}^{\text{dim}(\mathfrak{g})} \chi_t^{(j)}( \mathbf{e}_{\text{snap}})_t. 
\end{equation}
Hence, given $\text{dim}(\mathfrak{g})+1$ gradients $C_j$, we can construct $\text{dim}(\mathfrak{g})$ simultaneous equations 

\begin{equation}\label{eqn:grad-in-exp-snap-expanded}
    C_j - R_j \cdot C_1 = \sum_{t=1}^{\text{dim}(\mathfrak{g})}\Big( \chi_t^{(j)}( \mathbf{e}_{\text{snap}})_t - R_j \chi_t^{(1)}( \mathbf{e}_{\text{snap}})_t \Big) = 0 .
\end{equation}
Hence, we have a system of equations that can be solved as before, although as we can only technically find the ratios between gradients we will require $\text{dim}(\mathfrak{g})+1$ gradients as opposed to $\text{dim}(\mathfrak{g})$ gradients.

\section{High Frequency Trigonometric Models}\label{app:trig_input}
 
The Fourier series picture was utilized in \cite{kumar2023expressive} to analyze the privacy of VQC models. In this section, we examine this picture with regards to privacy analysis.

Quantum models admit a natural decomposition into a Fourier series, which means that if the model is constructed such that it contains high-frequency terms, then it gains certain privacy advantages  \cite{kumar2023expressive}. This is not a uniquely quantum effect however, and in this section, we shall show how it can be recreated with a classical linear classifier equipped with a trigonometric feature map.

\subsection{Quantum Models are Linear Classifiers with Trigonometric Feature Maps}

The connection between Quantum models and Fourier series has previously informed the construction of certain dequantisation techniques \cite{Schreiber_surrogates, sweke2023potential}. Classical surrogate models are capable of being run on classical machines and are simultaneously able to match the output of a quantum model to within some error.  It is known that the model output of a QML model as a Fourier series \cite{Schuld_2021} in the following form

\begin{align}
    y(\boldsymbol\theta) = \sum_{\omega \in \Omega} A_{\omega} e^{i \omega x},
\end{align}
where the $A_\omega $ coefficients depend on the entire VQC architecture, but the frequency spectrum $\Omega$ depends only on the encoding architecture. If we define
\begin{equation}
    b_{ \omega}(\boldsymbol\theta) := A_{\omega}(\boldsymbol\theta)  + A_{-\omega} (\boldsymbol\theta),
\end{equation}
\begin{equation}
    d_{ \omega}(\boldsymbol\theta) := i(A_{\omega}(\boldsymbol\theta)  - A_{-\omega} (\boldsymbol\theta)).
\end{equation}

Then the model output can be written as
\begin{align}
    y(\boldsymbol\theta) = a_0 + \sum_{\omega \in \Omega} \big( b_{\omega} \cos{ (\omega x)} + d_{\omega}\sin{(\omega x)} \big).
\end{align}
We can rewrite the model as a linear regression \cite{landman2022classically} defined by the inner product between a vector of coefficients  $\boldsymbol A $ and a trigonometric feature map $\phi( x)$ as follows
\begin{align}
    y(\boldsymbol\theta) = \boldsymbol A \cdot \phi( x),
\end{align}
where 
\begin{equation}
    \boldsymbol A = \sqrt{| \Omega|}(a_0, b_{\boldsymbol \omega_1}, d_{\boldsymbol \omega_1} ... b_{\boldsymbol \omega_\delta}, d_{\boldsymbol \omega_\delta}),
\end{equation}
\begin{equation}
    \phi( x) = \frac{1}{\sqrt{| \Omega|}} \big( 1, \cos{( \omega_1 x)}, \sin{( \omega_1 x)},..., \cos{( \omega_\delta x)}, \sin{( \omega_\delta x)} \big).
\end{equation}
Hence, a VQC model can effectively be thought of as a linear regression on data that has been transformed by a trigonometric feature map  $\phi( x)$. However, a key difference is that the coefficients in this linear regression model are not directly varied in the quantum case, but the coefficients  $\boldsymbol A $  have a dependency on the $\theta$ values in the variational circuit and therefore evolve as the $\theta$ values are adjusted during training. 

\subsection{ Classical High Frequency Trigonometric Models }

We consider the case of a fully classical model with a classical attack model. In this regime, we encode the data using some feature map $\phi(x)$ and then perform linear regression on the resulting input. This can be considered as a single-layer neural network with a single output neuron with the identity of an activation. Note that in the quantum case, the corresponding $\phi(x)$ is exponentially large. If we maintained this condition here, then the classical model would not even be able to be evaluated (as even storing $\phi(x)$ would require exponential resources), and hence in some sense, we would achieve a trivial form of privacy via considering a model we do not have the resources to run in the first place. We therefore assume a polynomial number of $\phi(x)$ values in the classical model. This setup,while technically distinct, is reminiscent of the technique of Random Fourier Features \cite{landman2022classically} which been shown to be successful in creating classical surrogates of quantum models by utilizing a polynomial number of Fourier features.

The classical model output is defined by
\begin{align}
    y(\boldsymbol\theta) = \boldsymbol A \cdot \phi(x),
\end{align}
where in the classical case the coefficients $A$ are adjusted variationally directly, we are effectively performing a linear regression. We can write the gradients as
\begin{equation}
    C_j = \frac{\partial \texttt{Cost}}{\partial A_j} = -2(y_i - \mathbf{A} \cdot \phi( x)) \Big[  \phi( x) \Big]_j.
\end{equation}
This allows simple relations between all the parameters to be written using the fact that
\begin{equation}
    \frac{C_k}{C_l} = \frac{[  \phi( x) ]_k}{[  \phi( x) ]_l},
\end{equation}
which if we substitute into the normalisation condition $|\phi(x)|^2 = 1$. Then we obtain a closed-form solution
\begin{equation}
    [  \phi( x) ]_i = \frac{C_i}{\sum_j C_j^2}.
\end{equation}
Hence, it is easy to recover the $\phi(x)$ values from the gradients. If we have $\phi(x)$ then we can also simply invert the equation to find $x$ as if we know that $[  \phi(\boldsymbol x) ]_k = \cos(\omega_k x)$ then we can find
\begin{equation}
    x = \frac{1}{\omega_k} \cos ^{-1}([ \phi( x) ]_k).
\end{equation}
This demonstrates that the feature map in the classical case is both retrievable and invertible (up to periodicity unless we enforce the feature map to be injective) and hence there is no analytical privacy in the feature space.

 An extension of this would be to use a neural network with the feature map $\phi(x)$ for the inputs. As long as the $\phi(x)$ values can be recovered by considering the neural network equations, then it will be possible to invert his feature map. Trigonometric feature maps of this structure have found success in some classical machine learning applications, such as learning high-frequency functions in low-dimensional spaces \cite{tancik2020fourier}.

 \subsection{Quantum High Frequency Trigonometric Models}

In the quantum case, if one assumes an exponentially large $\phi(x)$ generated by the quantum variational model as was the case in \cite{kumar2023expressive}. Then this would not be a problem for a classical model  using a trigonometric feature map as the ratio of specific gradients could be taken such that,
\begin{equation}
    \frac{C_k}{C_{k+1}} = \frac{[  \phi( x) ]_k}{[  \phi( x) ]_{k+1}} = \frac{\sin (\omega_k x)}{\cos (\omega_k x)} = \tan (\omega_k x),
\end{equation}
which would still allow one to invert and find $x$ as
\begin{equation}
    x = \frac{1}{\omega_k} \cos ^{-1}\Big( \frac{C_k}{C_{k+1}} \Big),
\end{equation}
using only two gradient values. However, in a quantum variational circuit, we do not train the Fourier coefficients $A$ directly, but we train the gate parameters $\theta$ that themselves influence the coefficients $A$. Hence, we cannot perform the same technique to recover this ratio of feature vector elements.

If we knew the coefficients $A$ and $\frac{\partial \texttt{Cost}}{\partial A_j}$, then it would be an easy task. However, we only have access to the $\theta$ gradients that can be written

\begin{equation}
    C_j = \frac{\partial \texttt{Cost}}{\partial \theta_j} = -2(y_i - \boldsymbol A \cdot \phi( x)) \Big( \frac{\partial \boldsymbol A}{\partial \theta_j} \cdot \phi( x) \Big).
\end{equation}
A solution to this system of equations would require calculating the inner product between exponentially large vectors. Even if one attempts a similar trick to the classical case to eliminate the $\boldsymbol A \cdot \phi(\boldsymbol x)$ term this would yield 
\begin{equation}
    \frac{C_k}{C_l} = \frac{\Big( \frac{\partial \boldsymbol A}{\partial \theta_k} \cdot \phi( x) \Big)}{\Big( \frac{\partial \boldsymbol A}{\partial \theta_l} \cdot \phi( x) \Big)},
\end{equation}
which would still require calculating the inner product of exponentially large vectors (before even considering the challenge of finding the $\frac{\partial \boldsymbol A}{\partial \theta_j}$ terms).

Hence, even if the feature map $\phi(x)$ is easy to invert, we see that the privacy of quantum variational circuits in the trigonometric feature map space can derive from the fact $\phi(x)$ is not recoverable. This is analogous to the snapshot recovery discussed in the main work.

While both classical and quantum high frequency trigonometric models provide protection against machine learning attacks and analytical Attacks in the input space. The exponential nature of quantum models provides analytical privacy in the trigonometric feature vector space, while polynomial-sized classical models do not exhibit privacy in feature map space and are feature vector recoverable. This section has shown explicitly how an exponential number of frequencies provides protection in the quantum case, against analytical attacks, when considering the Fourier space. However, they may still be vulnerable to the snapshot recovery and inversion techniques specified in the main text.

\subsection{Results on Trainability of VQCs}\label{app:trainability}

One of the main trainability problems that plague VQCs is exponentially vanishing gradients, more commonly known as the barren plateau (BP) problem \cite{larocca2024review, mcclean2018barren, larocca2022diagnosing}. The BP problem has been characterized \cite{fontana2023adjoint, ragone2023unified} for a restricted class of VQCs known as Lie algebraic supported ansatz (LASA), which cover a wide variety of commonly used models. Hence, when looking for trainable quantum models, it is easiest to restrict the LASA setting, which is what is done in this paper. Furthermore, it can be shown that under a stronger yet potentially more reasonable definition of vanishing gradients for VQCs, a necessary condition for a LASA to avoid a BP is to have a dynamical Lie algebra (DLA) with a polynomial dimension. It has been conjectured that this claim is far-reaching, in the sense that models that avoid BPs evolve within spaces that are polynomially sized, in terms of the number of qubits \cite{cerezo2023does}.

The following theorem presents a closed form expression for the variance, over uniform parameter initialization, of the quantum circuit gradient for input $\mathbf{x}$, denoted $\textup{GradVar}_{\mathbf{x}}$.
\begin{theorem}[Variance of Gradient Theorem 2.9 \cite{fontana2023adjoint}]
\label{thm:gradient thm}
Consider ansatz $\mathbf{U}(\boldsymbol{\theta})$ with DLA $\g$ admitting a decomposition into simple ideals $\g_\alpha$ and its center $\mathfrak{c}$. Given the input state $\rho(\mathbf{x})$ and if the measurement operator $i\boldsymbol{O} \in \g$ (LASA condition), then the variance of the gradient for the classical input $\mathbf{x}$, denoted by $\textup{GradVar}_{\mathbf{x}}$, is given by,
    \begin{equation}\label{eqn:gradvar-fulldetail}
        \textup{GradVar}_{\mathbf{x}} = \sum_{\alpha}\frac{\|\mathbf{H_{\g_\alpha}}\|^2_K\|\mathbf{O}_{\g_\alpha}\|^2_F \|\rho_{\g_\alpha}(\mathbf{x})\|^2_F}{\textup{dim}(\g_\alpha)^2},
    \end{equation}
where the subscript $\alpha$ under the operators $\mathbf{H}$, $\mathbf{O}$ and $\rho$ denotes the orthogonal projection (under Frobenius inner product) onto the ideal $\g_\alpha \subseteq \g$, %
Further, $\|\mathbf{H}\|^2_K$ is the Killing norm \cite{somma2004nature} of the generator of the parameter $\theta$ with respect to the gradient which is computed, i.e., $e^{-\theta i\mathbf{H}}$ in the ansatz $\mathbf{U}(\boldsymbol{\theta})$. The Killing norm is defined to be the Frobenius norm of the operator $\textup{ad}_{i \mathbf{H}}$.
\end{theorem}

A VQC is defined to be trainable when gradients can be efficiently estimated.
\begin{definition}[Trainability of VQC] \label{def:trainable}
    The VQC is considered trainable for input $\mathbf{x}$ if GradVar satisfies the condition,
    \begin{equation}
        \textup{GradVar}_{\mathbf{x}} = \mathcal{O}\left(\frac{1}{\text{poly}(n)}\right),
    \end{equation}
    where $n$ is the number of qubits.

\end{definition}

\begin{fact}[DLAs are Reductive \cite{fontana2023adjoint, ragone2023unified}]
The DLA admits the following orthogonal (in Frobenius inner product) decomposition:
 \begin{align}
 \label{eqn:splitting}
\g = \bigoplus_{\alpha}\g_{\alpha} \oplus \mathfrak{c},
\end{align}
where each $\g_{\alpha} \subset \g$ is a simple ideal (a minimal Lie subalgebra satisfying $\forall \mathbf{H} \in \g, \forall \mathbf{K} \in \g_{\alpha}, [\mathbf{H}, \mathbf{K} ] \in \g_{\alpha}$) and $\mathfrak{c} \subset \g$ is the center of $\mathfrak{g}$.
\end{fact}

This expression reveals the explicit dependence that the gradient variance has on the DLA dimension in the LASA setting, which is the reasoning behind the conjectured necessity that the DLA dimension must be polynomial to avoid BPs. More specifically, we can immediately see from Theorem~\ref{thm:gradient thm} that as long as ${\text{dim}(\mathfrak{g})}$ scales polynomially in $n$, and the Frobenius norm of the projection of input state in the DLA as well as the measurement operator is not vanishingly small in $n$, the variance of the gradient does not decay exponentially in $n$ and thus leads to a trainable model.

Another source of untrainability of the model is when the cost function optimization landscape is swamped with spurious local minima which renders it unamenable to any efficient optimizer to find a good approximation to the optimal solution. As shown by \cite{anschuetz2022}, models in the \emph{underparameterized phase} characterized by trainable parameters fewer than the degrees of freedom of the system exhibit the above spurious local minima behavior. A phase transition occurs when the number of parameters scale with the system's degrees of freedom where the local minima concentrate around the global minimum rendering it more efficient for the optimizers to converge to a good approximate solution. This is referred to as the \emph{overparameterized phase}. The work by \cite{larocca2023theory} shows that overparameterization can be achieved with the number of model parameters scaling as the size of the DLA dimension. Thus for $\text{poly}(n)$ sized DLA, overparameterization requires polynomial ansatz depth, thus ensuring trainability. 
\end{document}